\documentclass[11pt]{article}
\usepackage[utf8]{inputenc} 
\usepackage[a4paper, margin=1in]{geometry}
\usepackage{mathtools}
\mathtoolsset{showonlyrefs=true}
\usepackage{amsthm, amssymb, bbm, mathrsfs}
\usepackage{bm}
\usepackage{subcaption}
\usepackage{graphicx}
\usepackage{overpic}
\usepackage{float}
\usepackage[normalem]{ulem}
\usepackage{soul}
\usepackage{amsfonts}
\usepackage[colorlinks=true, linkcolor=blue, urlcolor=cyan, citecolor=green]{hyperref}
\usepackage[usenames,dvipsnames]{xcolor}
\usepackage{comment}
\usepackage{algorithm2e}

\theoremstyle{plain}
\newtheorem{defn}{Definition}[section]
\newtheorem{lemma}{Lemma}[section]
\newtheorem{remark}{Remark}[section]

\newtheorem{theorem}{Theorem}[section]
\newtheorem{cor}{Corollary}[section]

\newtheorem{prop}{Proposition}[section]
\numberwithin{equation}{section}
\numberwithin{figure}{section}
\numberwithin{table}{section}

\newcommand{\E}{\mathbb{E}}

\newcommand{\BS}{\rm BS}

\newcommand{\var}{{\rm var}}

\newcommand{\beas}{\begin{eqnarray*}}
\newcommand{\eeas}{\end{eqnarray*}}
\newcommand{\bal}{\begin{align}}
\newcommand{\eal}{\end{align}}
\newcommand{\bas}{\begin{align*}}
\newcommand{\eas}{\end{align*}}
\newcommand{\bea}{\begin{eqnarray}}
\newcommand{\eea}{\end{eqnarray}}
\newcommand{\tmop}{\end{eqnarray}}
\newcommand{\ben}{\begin{enumerate}}
\newcommand{\een}{\end{enumerate}}

\newcommand{\ui}{\mathrm{i}}

\newcommand{\dm}{\diamond}

\newcommand{\cF}{\mathcal{F}}

\newcommand{\mF}{\mathbb{F}}

\newcommand{\tmF}{\tilde {\mathbb{F}}}

\newcommand{\cO}{\mathcal{O}}

\newcommand{\bi}{\begin{itemize}}
\newcommand{\ei}{\end{itemize}}
\newcommand{\beq}{\begin{equation}}
\newcommand{\eeq}{\end{equation}}
\newcommand{\bv}{\begin{verbatim}}
\newcommand{\ev}{\end{verbatim}}

\newcommand{\ee}[1]{ {\mathbb{E}\left[{#1}\right]}}

\newcommand{\angl}[1]{\left\langle{#1}\right\rangle}

\newcommand{\bE}{\mathbb{E}}
\newcommand{\bR}{\mathbb{R}}

\usepackage{tikz}
\newcommand{\tkz}{\tikzexternaldisable}
\usepackage{mhequ} 
\usetikzlibrary{snakes}
\usetikzlibrary{decorations}
\usetikzlibrary{positioning}
\usetikzlibrary{shapes}
\usetikzlibrary{external}
\makeatletter
\pgfdeclareshape{crosscircle}
{
  \inheritsavedanchors[from=circle] 
  \inheritanchorborder[from=circle]
  \inheritanchor[from=circle]{north}
  \inheritanchor[from=circle]{north west}
  \inheritanchor[from=circle]{north east}
  \inheritanchor[from=circle]{center}
  \inheritanchor[from=circle]{west}
  \inheritanchor[from=circle]{east}
  \inheritanchor[from=circle]{mid}
  \inheritanchor[from=circle]{mid west}
  \inheritanchor[from=circle]{mid east}
  \inheritanchor[from=circle]{base}
  \inheritanchor[from=circle]{base west}
  \inheritanchor[from=circle]{base east}
  \inheritanchor[from=circle]{south}
  \inheritanchor[from=circle]{south west}
  \inheritanchor[from=circle]{south east}
  \inheritbackgroundpath[from=circle]
  \foregroundpath{
    \centerpoint%
    \pgf@xc=\pgf@x%
    \pgf@yc=\pgf@y%
    \pgfutil@tempdima=\radius%
    \pgfmathsetlength{\pgf@xb}{\pgfkeysvalueof{/pgf/outer xsep}}%
    \pgfmathsetlength{\pgf@yb}{\pgfkeysvalueof{/pgf/outer ysep}}%
    \ifdim\pgf@xb<\pgf@yb%
      \advance\pgfutil@tempdima by-\pgf@yb%
    \else%
      \advance\pgfutil@tempdima by-\pgf@xb%
    \fi%
    \pgfpathmoveto{\pgfpointadd{\pgfqpoint{\pgf@xc}{\pgf@yc}}{\pgfqpoint{-0.707107\pgfutil@tempdima}{-0.707107\pgfutil@tempdima}}}
    \pgfpathlineto{\pgfpointadd{\pgfqpoint{\pgf@xc}{\pgf@yc}}{\pgfqpoint{0.707107\pgfutil@tempdima}{0.707107\pgfutil@tempdima}}}
    \pgfpathmoveto{\pgfpointadd{\pgfqpoint{\pgf@xc}{\pgf@yc}}{\pgfqpoint{-0.707107\pgfutil@tempdima}{0.707107\pgfutil@tempdima}}}
    \pgfpathlineto{\pgfpointadd{\pgfqpoint{\pgf@xc}{\pgf@yc}}{\pgfqpoint{0.707107\pgfutil@tempdima}{-0.707107\pgfutil@tempdima}}}
  }
}
\makeatother

\def\X{\tikz[baseline=-2.8,scale=0.15]{\node[X] {};}} 
\def\M{\tikz[baseline=-2.8,scale=0.15]{\node[M] {};}} 

\def\XXd{\tikz[baseline=-1,scale=0.15]{\draw (-1,1) node[X] {} -- (0,0) node[not] {} -- (1,1) node[X] {};}} 

\def\MXd{\tikz[baseline=-1,scale=0.15]{\draw (-1,1) node[M] {} -- (0,0) node[not] {} -- (1,1) node[X] {};}} 

\def\MMd{\tikz[baseline=-1,scale=0.15]{\draw (-1,1) node[M] {} -- (0,0) node[not] {} -- (1,1) node[M] {};}} 

\def\MXdXd{\tikz[baseline=-1,scale=0.15]{
\draw (0,0) node[not] {} -- (-1,1) node[not] {}
-- (-2,2) node[M]{} ;
\draw (0,0) -- (1,1) node[X] {};
\draw (-1,1) -- (0,2) node[X] {};
}}

\def\MMdXd{\tikz[baseline=-1,scale=0.15]{
\draw (0,0) node[not] {} -- (-1,1) node[not] {}
-- (-2,2) node[M]{} ;
\draw (0,0) -- (1,1) node[X] {};
\draw (-1,1) -- (0,2) node[M] {};
}}

\def\MXdXdXd{\tikz[baseline=-1,scale=0.15]{
\draw (0,0) node[not] {} -- (-1,1) node[not] {}
-- (-2,2) node[not]{}  -- (-3,3) node[M]{};
\draw (0,0) -- (1,1) node[X] {};
\draw (-1,1) -- (0,2) node[X] {};
\draw (-2,2) -- (-1,3) node[X] {};
}}

\def\MXdMd{\tikz[baseline=-1,scale=0.15]{
\draw (0,0) node[not] {} -- (-1,1) node[not] {}
-- (-2,2) node[M]{} ;
\draw (-1,1) -- (0,2) node[X] {};
\draw (0,0) -- (1,1) node[M] {};
}}

\def\MMdMd{\tikz[baseline=1,scale=0.15]{
\draw (0,0) node[not] {} -- (-1,1) node[not] {}
-- (-2,2) node[M]{} ;
\draw (0,0) -- (1,1) node[M] {};
\draw (-1,1) -- (0,2) node[M] {};
}}

\def\MMdXdXd{\tikz[baseline=1,scale=0.15]{
\draw (0,0) node[not] {} -- (-1,1) node[not] {}
-- (-2,2) node[not]{}  -- (-3,3) node[M]{};
\draw (-2,2) -- (-1,3) node[M] {};
\draw (-1,1) -- (0,2) node[X] {};
\draw (0,0) -- (1,1) node[X] {};
}}

\def\MXdMdXd{\tikz[baseline=1,scale=0.15]{
\draw (0,0) node[not] {} -- (-1,1) node[not] {}
-- (-2,2) node[not]{}  -- (-3,3) node[M]{};
\draw (-2,2) -- (-1,3) node[X] {};
\draw (-1,1) -- (0,2) node[M] {};
\draw (0,0) -- (1,1) node[X] {};
}}

\def\MXdXdMd{\tikz[baseline=1,scale=0.15]{
\draw (0,0) node[not] {} -- (-1,1) node[not] {}
-- (-2,2) node[not]{}  -- (-3,3) node[M]{};
\draw (-2,2) -- (-1,3) node[X] {};
\draw (-1,1) -- (0,2) node[X] {};
\draw (0,0) -- (1,1) node[M] {};
}}

\def\MXdMXdd{\tikz[baseline=1,scale=0.15]{\draw (0,0) node[not] {} -- (-1,1) node[not] {};
\draw (0,0) -- (1,1) node[not] {};
\draw (-1,1) -- (-1.5,2.5) node[M] {};
\draw (-1,1) -- (-0.5,2.5) node[X] {};
\draw (1,1) -- (0.5,2.5) node[M] {};
\draw (1,1) -- (1.5,2.5) node[X] {};}}

\def\MXdXdXdXd{\tikz[baseline=1,scale=0.15]{
\draw (0,0) node[not] {} -- (-1,1) node[not] {}
-- (-2,2) node[not]{} -- (-3,3) node[not]  {} 
-- (-4,4) node[M]{};
\draw (-3,3) -- (-2,4) node[X] {};
\draw (-2,2) -- (-1,3) node[X] {};
\draw (-1,1) -- (0,2) node[X] {};
\draw (0,0) -- (1,1) node[X] {};
}}

\def\MXdXdXdXdXd{\tikz[baseline=3,scale=0.15]{
\draw (0,0) node[not] {} -- (-1,1) node[not] {}
-- (-2,2) node[not]{} -- (-3,3) node[not]  {} 
-- (-4,4) node[not]{}-- (-5,5) node[M]{};
\draw (-4,4) -- (-3,5) node[X] {};
\draw (-3,3) -- (-2,4) node[X] {};
\draw (-2,2) -- (-1,3) node[X] {};
\draw (-1,1) -- (0,2) node[X] {};
\draw (0,0) -- (1,1) node[X] {};
}}

\def\MMdXdXdXd{\tikz[baseline=3,scale=0.15]{
\draw (0,0) node[not] {} -- (-1,1) node[not] {}
-- (-2,2) node[not]{} -- (-3,3) node[not]  {} 
-- (-4,4) node[M]{};
\draw (-3,3) -- (-2,4) node[M] {};
\draw (-2,2) -- (-1,3) node[X] {};
\draw (-1,1) -- (0,2) node[X] {};
\draw (0,0) -- (1,1) node[X] {};
}}

\def\MXdMdXdXd{\tikz[baseline=3,scale=0.15]{
\draw (0,0) node[not] {} -- (-1,1) node[not] {}
-- (-2,2) node[not]{} -- (-3,3) node[not]  {} 
-- (-4,4) node[M]{};
\draw (-3,3) -- (-2,4) node[X] {};
\draw (-2,2) -- (-1,3) node[M] {};
\draw (-1,1) -- (0,2) node[X] {};
\draw (0,0) -- (1,1) node[X] {};
}}

\def\MXdXdMdXd{\tikz[baseline=3,scale=0.15]{
\draw (0,0) node[not] {} -- (-1,1) node[not] {}
-- (-2,2) node[not]{} -- (-3,3) node[not]  {} 
-- (-4,4) node[M]{};
\draw (-3,3) -- (-2,4) node[X] {};
\draw (-2,2) -- (-1,3) node[X] {};
\draw (-1,1) -- (0,2) node[M] {};
\draw (0,0) -- (1,1) node[X] {};
}}

\def\MXdXdXdMd{\tikz[baseline=3,scale=0.15]{
\draw (0,0) node[not] {} -- (-1,1) node[not] {}
-- (-2,2) node[not]{} -- (-3,3) node[not]  {} 
-- (-4,4) node[M]{};
\draw (-3,3) -- (-2,4) node[X] {};
\draw (-2,2) -- (-1,3) node[X] {};
\draw (-1,1) -- (0,2) node[X] {};
\draw (0,0) -- (1,1) node[M] {};
}}

\def\MMdXdMd{\tikz[baseline=3,scale=0.15]{
\draw (0,0) node[not] {} -- (-1,1) node[not] {}
-- (-2,2) node[not]{}  -- (-3,3) node[M]{};
\draw (-2,2) -- (-1,3) node[M] {};
\draw (-1,1) -- (0,2) node[X] {};
\draw (0,0) -- (1,1) node[M] {};
}}

\def\MMdMdXd{\tikz[baseline=3,scale=0.15]{
\draw (0,0) node[not] {} -- (-1,1) node[not] {}
-- (-2,2) node[not]{}  -- (-3,3) node[M]{};
\draw (-2,2) -- (-1,3) node[M] {};
\draw (-1,1) -- (0,2) node[M] {};
\draw (0,0) -- (1,1) node[X] {};
}}

\def\MXdMdMd{\tikz[baseline=3,scale=0.15]{
\draw (0,0) node[not] {} -- (-1,1) node[not] {}
-- (-2,2) node[not]{}  -- (-3,3) node[M]{};
\draw (-2,2) -- (-1,3) node[X] {};
\draw (-1,1) -- (0,2) node[M] {};
\draw (0,0) -- (1,1) node[M] {};
}}

\def\MMdMXdd{\tikz[baseline=3,scale=0.15]{\draw (0,0) node[not] {} -- (-1,1) node[not] {};
\draw (0,0) -- (1,1) node[not] {};
\draw (-1,1) -- (-1.5,2.5) node[M] {};
\draw (-1,1) -- (-0.5,2.5) node[M] {};
\draw (1,1) -- (0.5,2.5) node[M] {};
\draw (1,1) -- (1.5,2.5) node[X] {};}}

\def\MXdMXddXd{\tikz[baseline=3,scale=0.15]{
\draw (1,0) node[not] {} -- (0,1) node[not] {};
\draw (0,1) node[not] {} -- (-1,2) node[not] {};
\draw (0,1) -- (1,2) node[not] {};
\draw (-1,2) -- (-1.5,3.5) node[M] {};
\draw (-1,2) -- (-0.5,3.5) node[X] {};
\draw (1,2) -- (0.5,3.5) node[M] {};
\draw (1,2) -- (1.5,3.5) node[X] {};
\draw (1,0) node[not] {} -- (2,1) node[X] {};
}}

\def\MXdXdMXdd{\tikz[baseline=3,scale=0.15]{
\draw (1,0) node[not] {} -- (0,1) node[not] {};
\draw (0,1) node[not] {} -- (-1,2) node[not] {};
\draw (0,1) -- (0.5,2.5) node[X] {};
\draw (-1,2) -- (-1.5,3.5) node[M] {};
\draw (-1,2) -- (-0.5,3.5) node[X] {};
\draw (2,1) -- (1.5,2.5) node[M] {};
\draw (2,1) -- (2.5,2.5) node[X] {};
\draw (1,0) node[not] {} -- (2,1) node[not] {};
}}

\colorlet{symbols}{blue!90!black}
\colorlet{testcolor}{green!60!black}
\colorlet{connection}{red!30!black}

\tikzset{
root/.style={circle,fill=black!50,inner sep=0pt, minimum size=3mm},
        dot/.style={circle,fill=black,inner sep=0pt, minimum size=1.5mm},
        A/.style={very thin,circle,fill=BlueGreen!100,draw=black,inner sep=0pt,minimum size=1.2mm}, 
        X/.style={very thin,circle,fill=Gray!80,draw=black,inner sep=0pt,minimum size=1.2mm}, 
        Z/.style={very thin,circle,fill=RubineRed!100,draw=black,inner sep=0pt,minimum size=1.2mm}, 
        M/.style={very thin,circle,fill=YellowOrange!100,draw=black,inner sep=0pt,minimum size=1.2mm}, 
not/.style={thin,circle,fill=symbols,draw=connection,fill=connection,inner sep=0pt,minimum size=0.35mm},
	>=stealth,
        }

\begin{document}
\title{\textbf{Magic strikes for variance and gamma contracts, and other attainable claims}}
\author{Florian Bourgey\footnote{NYU, fb2615@nyu.edu}, Jim Gatheral\footnote{Baruch College, CUNY, jim.gatheral@baruch.cuny.edu}}
\date{\today}
\maketitle
\begin{abstract}
Building on the forest expansion of Al\`os, Gatheral and Radoi\v{c}i\'c
and on the explicit Bergomi--Guyon smile expansion derived by
Bourgey and Gatheral (2026), we derive fixed-point
approximations, in terms of the implied total variance at a small number of
\emph{magic strikes}, for the fair values of power payoff contracts; the
variance and gamma contracts arise as the endpoints of this one-parameter
family, priced by a single set of formulae.  We prove that such approximations
exist at every order: at order $n$, $2\lceil n/2\rceil+1$ magic strikes
suffice.
Numerical tests under the Heston and rough Bergomi models demonstrate good
accuracy even for strongly skewed smiles, and the method applies directly to
interpolated market smiles without strike inversion.  For the volatility contract, we show that the maturity-$T$
smile does not determine the leading correction to the Rolloos--Arslan approximation.
\end{abstract}



\section{Introduction}
For concreteness, we focus on forward variance models of the form
\begin{align}
\frac{dS_t}{S_t} &= \sqrt{V_t}\left(\rho\, dW_t + \sqrt{1-\rho^2}\, dW^\perp_t\right),\notag\\
d\xi_t(u) &= f_t(\xi)\, \kappa(u-t)\, dW_t,
\label{eq:SVmodel}
\end{align}
where $\xi_t(t) = V_t$, $X_t = \log S_t$, $V_t\,dt = d\langle X\rangle_t$, $W$ and $W^\perp$ are
two independent Brownian motions, and $\rho \in [-1,1]$. In particular, such a model is
scale-invariant, with $\xi$ adapted to the filtration generated by $W$.
For ease of notation and without loss of generality, suppose that $t=0$ and that the
entire implied volatility smile $\sigma_{\BS}(k,T)$ is known for a fixed expiration $T$,
where $k = \log(K/S_0)$ denotes the log-strike. Define the functions
$$
d_\pm(k) := \frac{-k}{\sigma_{\BS}(k,T)\sqrt{T}} \pm \frac{\sigma_{\BS}(k,T)\sqrt{T}}{2},
$$
and, following Fukasawa \cite{fukasawa2012normalizing}, denote their inverses by
$g_\pm(z) = d_\pm^{-1}(z)$. We further define
$$
\sigma_\pm(z) := \sigma_{\BS}(g_\pm(z),T)\sqrt{T},
$$
and $X_{t,T} := X_T - X_t = \log(S_T/S_t)$. We write $\langle X\rangle_{t,T} := \langle X\rangle_T-\langle X\rangle_t$ for the quadratic variation over $[t,T]$. The well-known Chriss--Morokoff formula
\cite{chriss1999market,gatheral2006volatility} gives the value of a variance contract as
\begin{equation}
M(T) := \E\left[\langle X\rangle_{0,T}\right] = -2\,\E\left[X_{0,T}\right]
= \int_0^T \xi_0(u)\, du = \int_\bR dz\, N'(z)\, \sigma_-^2(z),
\label{eq:VS}
\end{equation}
where $N(z) := \int_{-\infty}^z \frac{e^{-t^2/2}}{\sqrt{2\pi}}\, dt$ denotes the standard
normal cumulative distribution function (CDF). Similarly, as shown in
\cite{fukasawa2012normalizing}, the value of a gamma contract is
\begin{equation}
G(T) := \E\left[\frac{S_T}{S_0}\,\langle X\rangle_{0,T}\right]
= 2\,\E\left[X_{0,T}\, e^{X_{0,T}}\right] = \int_\bR dz\, N'(z)\, \sigma_+^2(z).
\label{eq:GS}
\end{equation}
The right-hand sides of \eqref{eq:VS} and \eqref{eq:GS} -- which we refer to as the
{\em Fukasawa representations} of the variance and gamma contracts -- express the fair
value as a Gaussian-weighted integral of the smile in the $z$-coordinate. They are exact
up to interpolation and extrapolation but require numerical strike inversion.
In practice, we do not observe the entire implied volatility smile: only a discrete set
of implied volatilities is available, which must be interpolated and extrapolated to
obtain the full smile. From \eqref{eq:VS} and \eqref{eq:GS} we have the simple
leading-order approximations
\begin{equation}
M(T) \approx \sigma_-(0)^2; \qquad G(T) \approx \sigma_+(0)^2.
\label{eq:approx0}
\end{equation}
These approximations correspond to approximating the variance and gamma contracts
using the implied volatility at log-strikes satisfying $d_-(k)=0$ and $d_+(k)=0$,
respectively. 
A related traders' rule, justified by Rolloos and Arslan
\cite{rolloos2017taylor} and analyzed by Al\`os, Rolloos and Shiraya
\cite{alos2021difference}, approximates the volatility contract
$\E\big[\sqrt{\langle X\rangle_{0,T}}\big]$ by $\sigma_-(0)$, the implied
volatility at the strike where the Black--Scholes vanna vanishes.
All of these approximations, however, are exact only to first order in
the volatility of volatility, and lose accuracy for the strongly skewed
smiles observed in practice. 
We show that the variance and gamma contract values can be recovered much
more accurately 
using the total implied variance evaluated at a small number of special
strikes, which we call {\em magic strikes}.  The strikes where $d_\mp(k)=0$ are the
first-order magic strikes for the variance and gamma contracts.

\subsubsection*{Our contributions}
First, building on the forest expansion of
\cite{alos2020exponentiation,friz2020forests} and on the Bergomi--Guyon (BG) smile
expansion \cite{bergomi2012stochastic}, computed explicitly to sixth order in
the companion paper \cite{bourgey2026demystifying}, we derive a systematic
framework that converts the smile expansion into fixed-point equations for
magic strikes, obtaining explicit higher-order approximations for the fair
values of power payoff contracts.
The variance and gamma
contracts are treated as the endpoints of a one-parameter family of power payoffs
priced by a single set of formulae, avoiding the strike inversion appearing in
the Fukasawa representations.
Second, we prove that the construction succeeds at every order, for any
attainable claim.
For the volatility contract, the second-order correction to the
Rolloos--Arslan approximation vanishes when $\rho=0$ but is otherwise not
determined by the smile.
Numerically, the method is fast and performs well in the Heston and rough
Bergomi models and on SPX market data.  Because the smile is read only at
the magic strikes, all within two standard deviations of $k_0$, the
approximations are independent of how the smile is extrapolated beyond
them, in contrast with representations that integrate over
the entire extrapolated smile.

The rest of the paper is organized as follows. In Section \ref{sec:forests}, we recall
the forest expansion \cite{alos2020exponentiation,friz2020forests} and the resulting
smile expansion, computed explicitly in \cite{bourgey2026demystifying}. In Section \ref{sec:magic-strikes}, we construct higher-order magic-strike approximations
for the variance, gamma, and power-payoff contracts.
In Section \ref{sec:iff}, we prove that magic-strike approximations exist
at every order for any attainable claim.  For the volatility contract, we
show that the maturity-$T$ smile does not determine the leading
correction to the Rolloos--Arslan approximation.  Section \ref{sec:numerics}
reports numerical tests under Heston and rough Bergomi dynamics, as well as on SPX market
data. Section \ref{sec:conclusion} concludes.

Code and tutorial notebooks accompanying this paper are available in the
GitHub repository \url{https://github.com/fbourgey/bergomi-guyon/}.

\section{The forest expansion}\label{sec:forests}
Following \cite{alos2020exponentiation}, for two continuous semimartingales $A$ and $B$ with integrable covariation process $\langle A, B\rangle$, the \textit{diamond product} is defined by $(A \dm B)_t(T) := \E[\langle A,B\rangle_{t,T} | \cF_t] = \bE[\langle A,B\rangle_{T} | \cF_t] - \langle A,B\rangle_{t}$. When $t=0$, we write $(A \dm B)(T) := (A \dm B)_0(T)$. From \cite[Corollary 3.1]{alos2020exponentiation}, we have the following result.
\begin{cor}\label{cor:CGF}
The cumulant generating function (CGF) is given by
\beq
\psi(T;u) := \log \ee{e^{\ui u X_{0,T}}} = - \frac{1}{2}u\,(u+\ui)\,M(T) + \sum_{\ell=1}^\infty\,\tilde {\mF}_\ell(u),
\label{eq:CGFmF}
\eeq
where the quantities $\tmF_\ell$ satisfy the recursion 
\tkz $\tmF_0 = -\tfrac 1 2 u\,(u+\ui)\,M = -\tfrac 1 2 u\,(u+\ui)\,\M$
\tikzexternaldisable
and for $\ell>0$,
\beq
\tmF_\ell=\frac12\,\sum_{j=0}^{\ell-2}\,\bigl(\tmF_{\ell-2-j} \dm \tmF_j\bigr) +\ui\, u\, \bigl(X \dm \tmF_{\ell-1}\bigr).
\label{eq:tmFrecursion}
\eeq
\end{cor}
\noindent In particular, \tkz
\beas
\psi(T;u-\ui/2)= - \tfrac{1}{2}(u^2+\tfrac14)\,M(T) + \sum_{\ell=1}^\infty\,\tilde {\mF}_\ell(u-\ui/2).
\label{eq:CGFmFm12}
\eeas
Note also that each tree in the forest $\tmF_\ell$ has $\ell+2$ leaves, where $X = \X$ is counted as a single leaf and $M= \M = \XXd$ is counted as a double leaf.  It is clear from the recursion \eqref{eq:tmFrecursion} that each such tree picks up a factor $- \tfrac{1}{2}(u^2+\tfrac14)$ for each $\M$ leaf and a factor $\ui (u-\tfrac {\ui}2) = \ui\, u +\tfrac12$ for each $\X$ leaf.  Since every tree has at least one $\M$ leaf, the final expression for the forest has an overall factor of $- \tfrac{1}{2}(u^2+\tfrac14)$.  In our graphical notation, the root of each subtree represents the diamond operation; for example, $\MXd$ represents $M\dm X$ while $\MXdXd$ represents $(M \dm X) \dm X$.
\subsection{Forests}\label{sec:forests6}
We list here trees appearing in the first few forests, up to order $\ell=5$, with their symmetry factors: 
\beas
 \tmF_0&:& \M\\
 \tmF_1&:& \MXd\\
\tmF_2&:& \left \{\MXdXd,\tfrac12\,\MMd\right\}\\
\tmF_3&:& \left \{\MXdXdXd ,\tfrac12\,\MMdXd ,\MXdMd \right\}\\
\tmF_4&:& \left \{ \MXdXdXdXd, \tfrac
12\,\MMdXdXd, \MXdMdXd,
\MXdXdMd, \tfrac12\,\MMdMd, \tfrac12\,\MXdMXdd  \right\}\\
\tmF_5&:& \left \{ \MXdXdXdXdXd, \tfrac12\,\MMdXdXdXd, \MXdMdXdXd, \MXdXdMdXd, \tfrac12\,\MMdMdXd, \tfrac12\,\MXdMXddXd,
\MXdXdXdMd, \tfrac12\,\MMdXdMd, \MXdMdMd, \MXdXdMXdd, \tfrac12\,\MMdMXdd  \right\}.
\eeas

\subsection{The smile}
Regarding the forest expansion \eqref{eq:CGFmF} as a formal power series
in $\epsilon$ whose power counts the forest index $\ell$, the implied
total variance smile $\Sigma(k):=\sigma_{\BS}^2(k,T)\,T$ has an expansion
of the same form,
\beq
\Sigma(k) =  \sum_{\ell=0}^\infty\,\epsilon^\ell\, a_\ell(k),
\label{eq:BG}
\eeq
analogous to the BG smile expansion
\cite{bergomi2012stochastic}, with
\tkz $a_0(k) = M(T)=\M$.
The mapping from the CGF to the coefficients $a_\ell(k)$ is carried out in
the companion paper \cite{bourgey2026demystifying}: each tree or product
of trees enters $a_\ell(k)$ with a model-independent polynomial
prefactor, and the $a_\ell$ are provided explicitly through order $6$ in
the Mathematica file \texttt{aBG.m}.  The number of trees grows rapidly
with the order: $\tmF_6$ already has $23$.
\begin{prop}\label{ass:degree}
In the smile expansion \eqref{eq:BG}, for every $\ell\ge1$, each
prefactor in $a_\ell(k)$, of a single tree or of a product of trees, is
a polynomial in $k$ of degree exactly $\ell$.
\end{prop}
\noindent
The proposition is proved in \cite{bourgey2026demystifying}; the
propositions of Section \ref{sec:magic-strikes} use $\ell\le5$, so at
every order used in this paper it may also be verified directly from
the Mathematica notebook \texttt{MagicStrikes.nb}.
\subsection{Explicit smile expansion through third order}
For concreteness, we reproduce from \cite{bourgey2026demystifying} the expansion through third order:
\tkz
\begin{equation}\label{eq:BG3}
\begin{aligned}
\Sigma(k)
&= 
\M +\epsilon  \left(\frac{k}{M}+\frac{1}{2}\right)\, \MXd 
+ 
\frac{1}{4} \epsilon^2 \left(\frac{k^2}{M^2}-\frac{1}{M}-\frac{1}{4}\right)\,\MMd 
\\
&\qquad+\epsilon^2 \left(\frac{k^2}{M^2}+\frac{k}{M}-\frac{1}{M}+\frac{1}{4}\right)\,\MXdXd 
+\frac{ \epsilon^2 }{4 M}\,\left(-\frac{5 k^2}{M^2}-\frac{2 k}{M}+\frac{3}{M}+\frac{1}{4}\right)\,(\MXd)^2
\\
&\qquad+\epsilon ^3 \left(\frac{k^3}{2 M^3}+\frac{k^2}{4 M^2}-\frac{3 k}{2 M^2}-\frac{k}{8 M}-\frac{1}{4 M}-\frac{1}{16}\right) \left[\tfrac{1}{2}\,\MMdXd+\MXdMd\right]
\\
&\qquad+ \epsilon ^3 \left(\frac{k^3}{M^3}+\frac{3 k^2}{2 M^2}-\frac{3 k}{M^2}+\frac{3 k}{4 M}-\frac{3}{2 M}+\frac{1}{8}\right)\,\MXdXdXd \\
&\qquad+\tfrac{1}{2} \epsilon ^3 \left(-\frac{2 k^3}{M^4}-\frac{k^2}{2 M^3}+\frac{7 k}{2 M^3}+\frac{k}{4 M^2}+\frac{1}{4 M^2}\right) \MMd\, \MXd
\\
&\qquad+\epsilon ^3 \left(-\frac{4 k^3}{M^4}-\frac{7 k^2}{2 M^3}+\frac{7 k}{M^3}-\frac{k}{2 M^2}+\frac{2}{M^2}+\frac{1}{8 M}\right) \MXd \,  \MXdXd
\\
&\qquad+\frac{1}{6}  \epsilon ^3 \left(\frac{39 k^3}{2 M^5}+\frac{45 k^2}{4 M^4}-\frac{24 k}{M^4}+\frac{3 k}{8 M^3}-\frac{9}{2 M^3}-\frac{3}{16 M^2}\right)\,(\MXd)^3 + \cO(\epsilon^4).
\end{aligned}
\end{equation}
Each diamond tree represents a model-dependent iterated integral which is in general hard to compute.  The polynomial prefactors are model-independent.
\section{Magic strikes}\label{sec:magic-strikes}
To first order in $\epsilon$, \eqref{eq:BG3} reads
$$
\Sigma(k) = \M +\epsilon  \Bigl(\frac{k}{M}+\frac{1}{2}\Bigr)\, \MXd +\cO(\epsilon^2).
$$
This gives the lowest-order approximation
\[
\Sigma\left(-\tfrac12 M \right) = M + \cO(\epsilon^2).
\]
To second order, it is straightforward to verify that
\[
    \frac 12\,\Sigma\left(-\tfrac12 M - \sqrt{M}\right)+\frac 12\,\Sigma\left(-\tfrac12 M + \sqrt{M}\right) = M - \frac {\epsilon^2}2 \frac{ (\MXd)^2}{M^2}+ \cO(\epsilon^3).
\]
\noindent Now, from \eqref{eq:BG3} once again,
\[
 \Sigma\left(-\tfrac12 M+\sqrt{M}\right)-\Sigma\left(-\tfrac12 M-\sqrt{M}\right)
 = \frac{2}{\sqrt{M}}\,\MXd\,\epsilon  + \cO(\epsilon^2).
\]
\noindent We have just proved the following proposition.
\begin{prop}
\begin{align}
&\frac 12\,\Sigma\left(-\tfrac12 M + \sqrt{M}\right)+\frac 12\,\Sigma\left(-\tfrac12 M - \sqrt{M}\right) \notag\\
&\qquad+\frac 1{8\,M}\,\left[\Sigma\left(-\tfrac12 M + \sqrt{M}\right)-\Sigma\left(-\tfrac12 M - \sqrt{M}\right)\right]^2 = M +\cO(\epsilon^3).
\label{eq:M2p}
\end{align}
\end{prop}
\begin{proof}
Following the argument above, or by direct expansion in $\epsilon$.
\end{proof}
\noindent A fixed-point iteration on \eqref{eq:M2p} gives an approximation of the value of the variance contract to second order in the forest expansion.  This approximation depends on only two strikes
$k_\pm := -\tfrac12 M \pm \sqrt{M}$, the second-order magic strikes.
\subsection{Higher-order expansions}
To go to a higher order, note that the BG smile expansion truncated at order $n$ is a polynomial in $k$ of degree $n$ (Proposition \ref{ass:degree}), of the form
\beq
\Sigma(k) = \sum_{j=0}^n\, \Sigma_j\,(k-k_0)^j + \cO(\epsilon^{n+1}),
\qquad
\Sigma_j := \frac{\Sigma^{(j)}(k_0)}{j!}.
\label{eq:taylor}
\eeq
Moreover, \eqref{eq:M2p} is of this form with the Taylor coefficients $\Sigma_j$, $j\in\{1,2\}$, replaced by finite difference approximations. Exploiting this observation,
by expressing polynomial expressions in $k$ by combinations of diamond trees, we can go to higher order than \eqref{eq:M2p}.
To do this, we need finite difference approximations of the derivatives $\Sigma^{(j)}(k_0)$ in terms of the total variance smile.
In particular, increasing the order should improve the approximation
precisely to the extent that the smile is well fitted, over the width of
the stencil, by a polynomial of the corresponding degree.

\begin{defn}\label{def:finite-differences}
Given a center $k_0$ and spacing $\Delta k>0$, define
\begin{equation}
\label{eq:dj}
\begin{aligned}
    d_0 &:= \Sigma(k_0),
    \\
    d_1 &:= \frac{1}{2\,\Delta k} \left[\Sigma(k_0+\Delta k) -\Sigma(k_0-\Delta k)\right],
    \\
    d_2 &:= \frac{1}{2\,\Delta k^2} \left[\Sigma(k_0+\Delta k) +\Sigma(k_0-\Delta k)-2\,\Sigma(k_0)\right],
    \\
    d_3 &:= \frac{1}{12\,\Delta k^3}
    \left[\Sigma(k_0+ 2\Delta k)- 2\,\Sigma(k_0 + \Delta k) + 2\,\Sigma(k_0 - \Delta k)
    -\Sigma(k_0 - 2 \Delta k)\right],
    \\
    d_4 &:= \frac{1}{12\,\Delta k^2}
    \left[
        \Sigma(k_0 + 2\Delta k) - 4\,\Sigma(k_0 + \Delta k) + 6\,\Sigma(k_0)
        - 4\,\Sigma(k_0 - \Delta k) + \Sigma(k_0 - 2\Delta k)
    \right].
\end{aligned}
\end{equation}
\end{defn}

The normalization of $d_4$ is intentional. For a quartic smile with Taylor coefficient $\Sigma_4$, the fourth central difference is $24(\Delta k)^4\Sigma_4$, and hence $d_4=2(\Delta k)^2\Sigma_4$, rather than $\Sigma_4$.

\begin{remark}
The magic strikes are $k_0+j\,\Delta k$, $j\in\{0,\pm1,\pm2\}$.
\end{remark}

\subsection{The variance contract}
\begin{prop}\label{prop:M}
Let $k_0 = -\tfrac 12 M$ and $\Delta k=\sqrt{M}$ in Definition \ref{def:finite-differences}. Then, we have the following expressions for $M$ up to fifth order in the forest expansion:
\begin{equation}
\label{eq:M12345}
\begin{aligned}
M & =d_0   + \cO(\epsilon^2)=:M_1+ \cO(\epsilon^2),
\\
M &  = M\,d_2 + \tfrac12\,d_1^2  +M_1 + \cO(\epsilon^3)=:M_2+ \cO(\epsilon^3),
\\
M & =-M\,d_2\,d_1 - \frac14\,d_1^3+M_2  + \cO(\epsilon^4) =: M_3+ \cO(\epsilon^4),
\\
M & =M\,d_4 +5 M\,d_3\,d_1 +2 M \,d_2^2 +\frac14 (10+3 M)\,d_2\,d_1^2 +
\frac 18 \,d_1^4 + M_3 + \cO(\epsilon^5)
=: M_4 + \cO(\epsilon^5),
\\
M & = -\frac52 M\, d_4 d_1 - 5 M^2\, d_3 \,d_2 -\frac{21}{2}M\, d_3 d_1^2 -10 M\, d_2^2 d_1
- \frac{1}{2}(9 +M)\,d_2 d_1^3\\
&\qquad\qquad\qquad\qquad\qquad\qquad\qquad\qquad\qquad\qquad- \tfrac{1}{16} d_1^5 + M_4 + \cO(\epsilon^6)
=: M_5 + \cO(\epsilon^6).
\end{aligned}
\end{equation}
\end{prop}
\begin{proof}
The result follows from direct expansion in
$\epsilon$.\footnote{The computation is straightforward but best carried out with a computer algebra system such as Mathematica\textsuperscript{\textregistered}.}
\end{proof}
\begin{remark}
Since $\Delta k^2=M$, the central value $\Sigma(k_0)$ cancels in
$d_0+M d_2$, so $M_2$ uses only the two strikes $k_0\pm\Delta k$.
The third-order correction reintroduces $\Sigma(k_0)$ through
$-M d_2d_1$, so $M_3$ also uses the central strike.
Thus $M_1,\ldots,M_5$ use one, two, three, five, and five strikes,
respectively, drawn from $k_0+j\,\Delta k$, $j\in\{0,\pm1,\pm2\}$:
$M_1$ uses $k_0$, while $M_4$ and $M_5$ use all five.
\end{remark}
\begin{remark}
The expression $M_1$ corresponds to the lowest-order approximation
$M(T)\approx\sigma_-(0)^2$ in \eqref{eq:approx0}; its square root is the
Rolloos--Arslan approximation of the volatility contract (Section
\ref{sec:convexity}).
\end{remark}
Fixed-point iterations of the expressions \eqref{eq:M12345} give approximations to the value of the variance contract to various orders in the forest expansion. Again, these approximations depend only on a very small number of (magic) strikes.
Although we could proceed to still higher orders, doing so typically introduces strikes outside the listed range.
\subsection{The gamma contract}
The gamma contract is the variance contract under the share measure (with $S$ as
num\'eraire).  At the level of the smile, this measure change is a reflection.
\begin{lemma}\label{lem:reflection}
Let $\Sigma^-(k):=\Sigma(-k)$ denote the reflected implied total variance smile.  Then
the gamma contract for $\Sigma$ is the variance contract for $\Sigma^-$:
\[
G[\Sigma]=M[\Sigma^-].
\]
\end{lemma}
\begin{proof}
We have
\[
d_-^{\,\Sigma^-}(k)=\frac{-k}{\sqrt{\Sigma(-k)}}-\frac{\sqrt{\Sigma(-k)}}{2}
=-\left[\frac{k}{\sqrt{\Sigma(-k)}}+\frac{\sqrt{\Sigma(-k)}}{2}\right]
=-\,d_+^{\,\Sigma}(-k).
\]
Inverting gives $g_-^{\,\Sigma^-}(z)=-g_+^{\,\Sigma}(-z)$, and therefore
\[
\sigma_-^{\,\Sigma^-}(z)^2=\Sigma^-\!\bigl(g_-^{\,\Sigma^-}(z)\bigr)
=\Sigma\bigl(g_+^{\,\Sigma}(-z)\bigr)=\sigma_+^{\,\Sigma}(-z)^2 .
\]
Substituting into \eqref{eq:VS} and using the symmetry of $N'$,
\[
M[\Sigma^-]=\int_\bR dz\,N'(z)\,\sigma_+^{\,\Sigma}(-z)^2
=\int_\bR dz\,N'(z)\,\sigma_+^{\,\Sigma}(z)^2=G[\Sigma].
\qedhere
\]
\end{proof}
We recall from \cite{alos2020exponentiation} that
the fair value of the gamma contract is given in terms of trees by
\[
G = M+ \sum_{\ell=1}^\infty\,X^{\dm \ell}M = \M + \MXd\,\epsilon + \MXdXd\,\epsilon^2 +\MXdXdXd\,\epsilon^3 +\MXdXdXdXd\,\epsilon^4 +\MXdXdXdXdXd\,\epsilon^5+\cO(\epsilon^6).
\]
Here $X^{\dm 0}M:=M$ and, recursively,
$X^{\dm(\ell+1)}M:=X\dm(X^{\dm\ell}M)$ for $\ell \in \mathbb{N}$.
\begin{prop}\label{prop:gamma}
Let $k_0=\tfrac12G$ and $\Delta k=\sqrt{G}$ in Definition \ref{def:finite-differences}. Then, we have the following expressions for $G$ up to fifth order in the forest expansion:
\begin{equation}
\label{eq:G12345}
\begin{aligned}
G &= d_0+\cO(\epsilon^2)=:G_1+\cO(\epsilon^2),
\\
G &= G \,d_2+\tfrac12 d_1^2+G_1+\cO(\epsilon^3)=:G_2+\cO(\epsilon^3),
\\
G &= G\, d_2d_1+\tfrac14d_1^3+G_2+\cO(\epsilon^4)=:G_3+\cO(\epsilon^4),
\\
G &= G\,d_4+5 G\,d_3d_1+2 G\,d_2^2
+\tfrac14 \left(10+3G\right)d_2d_1^2
+\tfrac18d_1^4+G_3+\cO(\epsilon^5)=:G_4+\cO(\epsilon^5),
\\
G &= \tfrac52 G d_4 d_1+5 G^2\, d_3 d_2+\tfrac{21}{2}G\,d_3d_1^2
+10 G\,d_2^2d_1
+\tfrac12\left(G+9\right)d_2d_1^3
\\
&\quad\quad\quad\quad\quad\quad\quad\quad\quad\quad\quad\quad
\quad\quad\quad\quad\quad+\tfrac{1}{16}d_1^5+G_4+\cO(\epsilon^6)=:G_5+\cO(\epsilon^6).
\end{aligned}
\end{equation}
\end{prop}
\begin{proof}
Apply Proposition \ref{prop:M} to the reflected smile $\Sigma^-$, whose variance
contract is $G$ by Lemma \ref{lem:reflection}.  The center and spacing prescribed
there are $k_0^-=-\tfrac12G$ and $\Delta k=\sqrt{G}$, and $k_0^-$ corresponds to
the log-strike $k_0=\tfrac12G$ for $\Sigma$, as in the statement.  Since
\[
\Sigma^-\bigl(k_0^-+\beta\,\Delta k\bigr)=\Sigma\bigl(k_0-\beta\,\Delta k\bigr),
\]
and since the stencils of Definition \ref{def:finite-differences} are symmetric in
$\beta\mapsto-\beta$ for $j$ even and antisymmetric for $j$ odd, the quantities $d_j^-$ formed
from $\Sigma^-$ at $k_0^-$ and those formed from $\Sigma$ at $k_0$ are related by
\[
d_j^-=(-1)^j\,d_j ,\qquad j=0,1,2,3,4 .
\]
Substituting $M\mapsto G$ and $d_j\mapsto(-1)^jd_j$ into \eqref{eq:M12345} gives
\eqref{eq:G12345}: a monomial $d_1^{\,a}d_2^{\,b}d_3^{\,c}d_4^{\,e}$ acquires the
factor $(-1)^{a+c}$, so the terms of odd total degree in $d_1$ and $d_3$ change sign
and all others are unchanged.
\end{proof}
\begin{remark}
The expression $G_1$ corresponds to the lowest-order approximation $G(T)  \approx \sigma_+(0)^2$ in  \eqref{eq:approx0}.
\end{remark}
\subsection{Why no stencil parameter?}\label{sec:whyz}
A natural generalization of Definition \ref{def:finite-differences} replaces
$\Delta k$ by $z\,\Delta k$ for a stencil parameter $z>0$, so that the magic
strikes become $k_0+j\,z\,\Delta k$.  Nothing breaks except that from order 4, the coefficients depend on
$z$, so the formulae become more complicated.
More important is the impact on the implied option portfolio.  At the flat smile
$\Sigma(k_j)\equiv M$, where the nonlinear terms vanish, the sensitivities
$\partial M_5/\partial\Sigma(k_j)$ reduce
to the unique weights on the five strikes $k_0+j\,z\,\Delta k$ that
reproduce the Gaussian average
$\E\bigl[f\bigl(k_0+\sqrt M\,Z\bigr)\bigr]$, $Z\sim N(0,1)$, on
polynomials $f$ of degree at most $5$.  Explicitly,
\[
w_0=1-\frac{5z^2-3}{4z^4},\qquad
w_{\pm1}=\frac{4z^2-3}{6z^4},\qquad
w_{\pm2}=\frac{3-z^2}{24z^4}.
\]
These weights are nonnegative if and only if
$z\in\bigl[\tfrac{\sqrt3}2,\sqrt3\bigr]$.  Inside this interval the approximation $M_5$
is, at leading order, a convex average of implied total variances;
outside, the conditioning degrades rapidly.
Lower orders tighten the interval.  At order 2 for example, 
\[
w_0=1-\frac{1}{z^2},\qquad
w_{\pm1}=\frac{1}{2z^2},
\]
nonnegative if and only if $z\ge1$. Nonnegative weights at every order through 5 therefore require
$z\in[1,\sqrt3\,]$. Choosing $z>1$ increases the chance that some magic strikes $k_0+j\,z\,\Delta k$ lie outside the listed range.  The
choice $z=1$  is thus the natural
choice for the stencil parameter.
\subsection{Power payoffs}\label{sec:power}
For $p\in[0,1]$, the power payoff $e^{p X_{0,T}}=(S_T/S_0)^p$ is
European, and \cite[eq.~(1.3)]{fukasawa2012normalizing} gives its value as an exact
Gaussian average of the smile, generalizing the Fukasawa representations
\eqref{eq:VS} and \eqref{eq:GS}; see also \cite{de2018moment}. For $p\in(0,1)$, define the \emph{implied power variance}
$W_p=W_p(T)$ by
\[
\ee{e^{p X_{0,T}}}=e^{\lambda_p \,W_p},
\qquad
\lambda_p:=\tfrac12\,p\,(p-1).
\]
At the endpoints where $\lambda_p=0$, define $W_p$ by continuous extension:
\begin{align*}
W_0&:=\lim_{p\downarrow0}\frac{\log\ee{e^{p X_{0,T}}}}{\lambda_p}
=-2\,\ee{X_{0,T}}=M,\\
W_1&:=\lim_{p\uparrow1}\frac{\log\ee{e^{p X_{0,T}}}}{\lambda_p}
=2\,\ee{X_{0,T} e^{X_{0,T}}}=G.
\end{align*}
For $p\in(0,1)$, substituting $\ui\,u=p$ into the CGF recursion
\eqref{eq:tmFrecursion} and dividing by $\lambda_p$ gives the implied power
variance in terms of trees:
\begin{align*}
W_p &= M + p\,\MXd\,\epsilon
+\Bigl(p^2\,\MXdXd+\tfrac12\lambda_p\,\MMd\Bigr)\epsilon^2
+\left[p^3\,\MXdXdXd+p\, \lambda_p\,\left(\tfrac 12 \MMdXd +
 \MXdMd\right)
\right]\,\epsilon ^3 \\
&\qquad
+ \left[p^4\,\MXdXdXdXd +p^2\, \lambda _p\, \left(\tfrac12\,\MMdXdXd+\MXdMdXd+\MXdXdMd+\tfrac12\MXdMXdd\right)+\tfrac{1}{2}\,\lambda_p^2\,\MMdMd
\right]\,\epsilon ^4\\
&\qquad\qquad+ \bigg[
p^5\,\MXdXdXdXdXd+
p^3\, \lambda_p \left(\tfrac12\,\MMdXdXdXd+\MXdMdXdXd+\MXdXdMdXd+\MXdXdMXdd+\MXdXdXdMd+\tfrac12\,\MXdMXddXd\right)\\
&\qquad\qquad\qquad\qquad\qquad\qquad\qquad\qquad+p \,\lambda_p^2 \left(\tfrac12\,\MMdMdXd+\tfrac12\,\MMdXdMd+\MXdMdMd+\tfrac12\,\MMdMXdd\right)\bigg]\,\epsilon ^5
+\cO(\epsilon^6).
\end{align*}
In particular, $W_0=M$ and, since $\lambda_p\to0$ kills every term with more than one
leaf $\M$, $W_1=G$.  The family thus interpolates the variance and gamma contracts
through European payoffs.
\begin{prop}\label{prop:power}
Let $k_0=(p-\tfrac12)\,W_p$ and $\Delta k=\sqrt{W_p}$ in Definition
\ref{def:finite-differences}, and write
\[
\omega:=\lambda_p\,W_p,\qquad q:=2p-1 .
\]
Then, with $W:=W_p$, we have the following expressions for $W_p$ up to fifth order
in the forest expansion:
\begin{equation}
\label{eq:W12345}
\begin{aligned}
W &= d_0+\cO(\epsilon^2)=:P_1+\cO(\epsilon^2),
\\
W &= W\,d_2+\tfrac12(1+\omega)\,d_1^2+P_1+\cO(\epsilon^3)=:P_2+\cO(\epsilon^3),
\\
W &= q\left[W\,d_2d_1+\tfrac14(1+\omega)\,d_1^3\right]+P_2+\cO(\epsilon^4)
=:P_3+\cO(\epsilon^4),
\\
W &= W\,d_4+(5+2\omega)\,W\,d_3d_1+(2+\omega)\,W\,d_2^2
+\left[\tfrac52+\tfrac34 W+\tfrac{21}2\,\omega+\omega^2\right]d_2d_1^2
\\
&\qquad\qquad\qquad
+\tfrac18\left[1+11\lambda_p+(10\lambda_p+1)\,\omega\right]d_1^4
+P_3+\cO(\epsilon^5)=:P_4+\cO(\epsilon^5),
\\
W &= q\,\Bigl[\tfrac52\,W\,d_4d_1+5\,W^2d_3d_2
+\bigl(\tfrac{21}{2}+\tfrac92\,\omega\bigr)W\,d_3d_1^2
+\bigl(10+\tfrac92\,\omega\bigr)W\,d_2^2d_1
\\
&\qquad\qquad
+\bigl(\tfrac92+\tfrac12 W+12\,\omega+\tfrac32\,\omega^2\bigr)d_2d_1^3
+\tfrac1{16}\bigl(1+17\lambda_p+(14\lambda_p+1)\,\omega\bigr)d_1^5\Bigr]\\
&\qquad\qquad\qquad\qquad\qquad\qquad\qquad\qquad\qquad\qquad
+P_4+\cO(\epsilon^6)=:P_5+\cO(\epsilon^6).
\end{aligned}
\end{equation}
\end{prop}
\begin{proof}
The result follows from direct expansion in
$\epsilon$.
\end{proof}
\begin{remark}
At the endpoints $\lambda_p=0$, so $\omega=0$.  When $p=0$, the equations
\eqref{eq:W12345} reduce to \eqref{eq:M12345}, and when $p=1$ to
\eqref{eq:G12345}.  The overall factor $q=2p-1$ on the odd orders is the reflection
of Lemma \ref{lem:reflection}, which 
flips the
sign of the odd stencils $d_1$ and $d_3$.
\end{remark}
\section{Magic strikes at every order}\label{sec:iff}
The propositions of Section \ref{sec:magic-strikes} exhibit magic-strike approximations up to fifth order for the variance, gamma and power-payoff contracts.  In this section we prove that the construction succeeds at every order for
European and other attainable claims (Theorem \ref{thm:exist}), and in
particular for power payoffs (Corollary \ref{cor:exist}).  
As an example where the magic-strike construction cannot succeed, we examine the volatility contract from the tree perspective (Section \ref{sec:convexity}).
\subsection{Computing the approximations}\label{sec:algorithm}
In practice the approximations of Section \ref{sec:magic-strikes} are found by
successive elimination, one order at a time, as summarized in Algorithm
\ref{algo:Fn} and illustrated in the Mathematica notebook \verb+MagicStrikes.nb+.
Since $d_j=\cO(\epsilon^j)$, with leading coefficient supported on the single trees of order $j$, the monomial targeted at each pass of the loop appears, among the products $d_{j_1}\cdots d_{j_r}$ whose coefficients are not yet determined, only in the current one, and with a nonzero multiplier.  Each pass therefore solves one linear equation for one coefficient, in terms of those already found.  In particular the coefficients in Propositions \ref{prop:M}, \ref{prop:gamma} and \ref{prop:power} are the unique ones, order by order, among linear combinations of the products $d_{j_1}\cdots d_{j_r}$ with $j_1+\cdots+j_r=n$.
\RestyleAlgo{boxruled}
\LinesNumbered
\begin{algorithm}[!h]
\caption{The order-$n$ magic-strike approximation, $n\ge2$}\label{algo:Fn}
Given $P_{n-1}$, with $P_1=d_0$, write
$P_n=P_{n-1}+\sum_{j_1+\cdots+j_r=n} c_{j_1\cdots j_r}\,d_{j_1}\cdots d_{j_r}$,
with one undetermined coefficient $c_{j_1\cdots j_r}$, polynomial in $W_p$, for
each partition of $n$.\\
Substituting the smile expansion \eqref{eq:BG}, expand $W_p-P_n$ at order
$\epsilon^n$ as a combination of trees and products of trees.\\
\For{each partition $j_1+\cdots+j_r=n$, coarsest first: $r=1$, then $r=2$,
and so on}{
Among the products of trees of orders $j_1,\dots,j_r$ in $W_p-P_n$, select the
one whose factors have the greatest number of $\X$ leaves, and choose
$c_{j_1\cdots j_r}$ to eliminate it.  Trees with the same numbers of leaves
$\M$ and $\X$ have identical prefactors, and are eliminated at the same
time.\\
}
Return $P_n$; the remaining tree monomials at order $\epsilon^n$ are
eliminated automatically, and $W_p-P_n=\cO(\epsilon^{n+1})$.
\end{algorithm}

\begin{remark}
The approximations read the smile only at the magic strikes; the smile
beyond two standard deviations of $k_0$ plays no role.  In effect, the
construction fits a polynomial to the smile at the magic strikes and
prices the contract as if the smile were that polynomial.  By
Proposition \ref{ass:degree} the smile is polynomial in $k$ at each
order in $\epsilon$, so nothing is lost to the order of the expansion;
the approximation is accurate because, and to the extent that, the
smile is locally polynomial.
\end{remark}

What the elimination does not by itself guarantee is the last line of Algorithm \ref{algo:Fn}: at each pass of the loop one coefficient is available, while several trees may require elimination.  It turns out that when the tree with the greatest number of leaves $\X$ is eliminated, the remaining trees of the same order are eliminated with it.  That this always happens is guaranteed by Theorem \ref{thm:exist} below.

\subsection{Existence}\label{sec:theorem}\label{sec:properties}

The following theorem shows that a magic-strike representation exists for
every European claim.
\begin{theorem}[Existence, European claims]\label{thm:exist}
Let the payoff $g$ be measurable with
$|g(S_0e^k)|\le C e^{a|k|}$ for all $k\in\mathbb R$, for some
$C>0$ and $a\ge0$. Assume $\ee{|g(S_T)|}<\infty$, and let $V=\ee{g(S_T)}$.
For every $n\ge1$, every $k_0$ and every $\Delta k>0$, there exists an
explicit function $F^g_n$ such that
\[
V=F^g_n\bigl(\Sigma(k_0+j\,\Delta k),\ |j|\le\lceil n/2\rceil\bigr)
+\cO(\epsilon^{n+1}) .
\]
\end{theorem}
\begin{proof}
By Breeden--Litzenberger, the density of $X_{0,T}$ is an explicit function of
the smile and its first two derivatives:
\beq
V=\int_\bR g\bigl(S_0\,e^{k}\bigr)\,q(k)\,dk,
\qquad
q(k)=\frac{\gamma(k)}{\sqrt{2\pi\,\Sigma(k)}}\;e^{-\frac12 d_-(k)^2},
\label{eq:BL}
\eeq
with
\[
\gamma(k)=\Bigl(1-\frac{k\,\Sigma'(k)}{2\,\Sigma(k)}\Bigr)^{2}
-\frac{\Sigma'(k)^2}{4}\Bigl(\frac{1}{\Sigma(k)}+\frac14\Bigr)
+\frac{\Sigma''(k)}{2} .
\]
By Proposition \ref{ass:degree}, $\Sigma$ agrees to $\cO(\epsilon^{n+1})$
with a polynomial in $k$ of degree $n$, determined by its Taylor
coefficients $\Sigma_j:=\Sigma^{(j)}(k_0)/j!$, $j\in\{0,\dots,n\}$, about
$k_0$; so then are $\Sigma'$, $\Sigma''$, and hence $q$.  At leading
order the smile is flat, $\gamma=1$, and $q$ is a Gaussian density;
expanding \eqref{eq:BL} in $\epsilon$, each term is that Gaussian times a
polynomial in $k$ whose coefficients are explicit in the $\Sigma_j$, and each coefficient is integrable against $g(S_0e^k)$, since a
Gaussian density times a polynomial remains integrable after
multiplication by $e^{a|k|}$.
Hence
\[
V=\widetilde F^g_n(\Sigma_0,\dots,\Sigma_n)+\cO(\epsilon^{n+1})
\]
for an explicit function $\widetilde F^g_n$.  Finally, the Taylor
coefficients of a polynomial of degree $n$ are linear combinations of its
values at any $n+1$ points; taking the points $k_0+j\Delta k$,
$|j|\le\lceil n/2\rceil$, $\widetilde F^g_n$ becomes the $F^g_n$ of the
statement.
\end{proof}

Theorem \ref{thm:exist} extends immediately to any \emph{attainable}
claim: one replicable by a fixed portfolio of European options together
with a hedging strategy that is model-free within the class
\eqref{eq:SVmodel}, so that its value coincides with that of a European
claim.  The variance contract, with value $M=-2\,\ee{X_{0,T}}$ by
Chriss--Morokoff \cite{chriss1999market}, is the classic example.
For the power payoffs, the strikes can be centered by the claim itself: choosing $k_0$ and $\Delta k$ in terms of $W_p$, both the strikes and the coefficients of $F_n$ depend on $W_p$, and the representation becomes a fixed-point equation for $W_p$, as in Section \ref{sec:magic-strikes}.

\begin{cor}[Existence, power payoffs]\label{cor:exist}\label{thm:iff}
Let $p\in[0,1]$ and let $W_p$ be the implied power variance.  For every
$n\ge1$ there exists an explicit function $F_n$, with coefficients
explicit in $W_p$, such that
\[
W_p \;=\; F_n\bigl(\Sigma(k_0+j\,\Delta k),\ |j|\le\lceil n/2\rceil\bigr)
\;+\;\cO(\epsilon^{n+1}),
\qquad
k_0=(p-\tfrac12)W_p,\quad \Delta k=\sqrt{W_p} .
\]
\end{cor}
\begin{proof}
Take $g(S)=(S/S_0)^p$, so that $V=\ee{e^{pX_{0,T}}}=e^{\lambda_p W_p}$.  In
the proof of Theorem \ref{thm:exist}, the payoff tilts the Gaussian by
$e^{pk}$, and
\[
\log\ee{e^{pX_{0,T}}}=\Phi_n(p;\Sigma_0,\dots,\Sigma_n)+\cO(\epsilon^{n+1}),
\]
with $\Phi_n$ polynomial in $p$ at each order.  
Since $\ee{1}=1$ and
$\ee{e^{X_{0,T}}}=1$, each retained order of $\Phi_n$ vanishes at $p=0$ and
$p=1$, so $\lambda_p=\tfrac12\,p\,(p-1)$ divides it, and
$W_p=\Phi_n/\lambda_p$ is well defined for every $p\in[0,1]$, endpoints
included.  Taking the reconstruction points $k_0+j\Delta k$ with
$k_0=(p-\tfrac12)W_p$ and $\Delta k=\sqrt{W_p}$, the coefficients depend
on $W_p$, and the resulting relation is the fixed-point equation.

\end{proof}
Corollary \ref{cor:exist} guarantees that the eliminations of Algorithm
\ref{algo:Fn}, which produced Propositions \ref{prop:M},
\ref{prop:gamma} and \ref{prop:power}, succeed at every order.
\subsubsection{On the choice of magic strikes}
Theorem \ref{thm:exist} makes no choice of strikes: any $n+1$ of them
reconstruct the polynomial, at the cost of coefficients that depend on
the strikes chosen.  What forces the choice made in Definition
\ref{def:finite-differences} is the match at order $\epsilon$ from a
single strike: comparing the order-$\epsilon$ term of $W_p$ with the
prefactor of $\MXd$ in \eqref{eq:BG3}, $P_1=d_0$ matches $W_p$ to
$\cO(\epsilon^2)$ precisely when $k_0=(p-\tfrac12)W_p$.

\subsubsection{Magic strikes for other claims}\label{sec:other_claims}
Theorem \ref{thm:exist} applies to any attainable claim.  The leverage
contract $L=G-M$, for example, is best approximated as the difference of
the magic-strike approximations of $G$ and $M$.  The stochasticity
contract $\zeta=\var[X_{0,T}]-\ee{\angl{X}_{0,T}}$\footnote{The three expressions for $\zeta$ agree.  Each tree enters
$W_p$ through the factor $\lambda_p^{m-1}p^{x}$, whose derivative at
$p=0$ vanishes except for $\MXd$ and $\tfrac12\,\MMd$, so that
$\partial_p W_p\big|_{p=0}=\MXd-\tfrac14\,\MMd$.  Similarly, only $\M$,
$\MXd$ and $\tfrac12\,\MMd$ contribute to the coefficient of $u^2$ in
the CGF \eqref{eq:CGFmF}, so that the second cumulant of $X_{0,T}$ is
$\var(X_{0,T})=M-\MXd+\tfrac14\,\MMd$, whence
$\zeta=\var(X_{0,T})-\ee{\angl{X}_{0,T}}=\tfrac14\,\MMd-\MXd
=-\,\partial_p W_p\big|_{p=0}$.} inherits its magic-strike
approximation from the power contract through the exact identity
$\zeta=-\partial_p W_p\big|_{p=0}$: differentiating the fixed-point
equation in $p$ introduces no new strikes.  In both cases the error is
comparable to that of the underlying power-contract approximations,
at worst the sum of two such errors; the claims themselves being small,
the relative error is correspondingly larger.

\subsection{The volatility contract}\label{sec:convexity}

The magic-strike construction does not extend to every claim.  The
volatility contract $\ee{\sqrt{\angl X_{0,T}}}$ is an obvious example.  The
traders' approximation justified by Rolloos and Arslan (RA)
\cite{rolloos2017taylor} approximates it by $\sigma_-(0)$, the square root
of the first-order magic-strike approximation of the variance contract.
This is exact to first order in the forest expansion, and to second order
when $\rho=0$, but the maturity-$T$ smile does not determine the leading
correction.  It is instructive to see why explicitly in tree language.

We start by expanding the volatility contract in diamond trees.  Note that only trees built
from $\M$ alone appear: the payoff $\sqrt{\angl {X}_{0,T}}$ involves only the quadratic
variation but not the underlying.
\begin{prop}\label{prop:volDiamond}
\begin{align}
\ee{\sqrt{\angl{X}_{0,T}}}
= \sqrt{M}\left\{1 -\frac{\epsilon^2}{8}\,\frac{\MMd}{M^2}+\frac{3\,\epsilon^4}{16}\,\frac{\MMdMd}{M^3}
- \frac{15\epsilon^4}{128}\frac{(\MMd)^2}{M^4}
+ \cO(\epsilon^6)
\right\}.
\label{eq:volDiamond}
\end{align}
\end{prop}
\begin{proof}
As in \cite[Chapter 11]{gatheral2006volatility}, the volatility contract admits the Laplace-transform representation
\begin{equation}
\ee{\sqrt{\angl{X}_{{0,T}}}}
= \frac{1}{2 \sqrt{\pi}} \int_0^\infty
\frac{1-\ee{e^{-\lambda \angl{X}_{0,T}}}}{\lambda^{3/2}}\, d\lambda.
\label{eq:ExpectedVol}
\end{equation}
From \cite[Theorem 1.1]{friz2020forests}, applied with $z_1=0$ and $z_2=-\lambda$ and with the formal expansion parameter $\epsilon$ tracking the forest order, the moment generating function of the quadratic variation satisfies
\begin{align}
\ee{e^{-\lambda\,\angl{X}_{0,T}}}
&= \exp\left\{-\lambda\,\M+\epsilon^2\,\frac{\lambda^2}{2}\,\MMd -\epsilon^4\,\frac{\lambda^3}{2}\,\MMdMd + \cO(\epsilon^6)\right\}.
\label{eq:mgf}
\end{align}
Substituting \eqref{eq:mgf} into \eqref{eq:ExpectedVol}, expanding the exponential in $\epsilon$, and integrating term by term in $\lambda$ yields \eqref{eq:volDiamond}.
\end{proof}
\noindent
On the other hand, the zero-vanna strike differs from $-\tfrac M2$ by
$\cO(\epsilon^2)$ and $\Sigma'=\cO(\epsilon)$, so
$\sigma_-(0)=\sqrt{\Sigma\left(-\tfrac M2\right)}+\cO(\epsilon^3)$, and
from \eqref{eq:BG3},
\beq
\sqrt{\Sigma\left(-\tfrac M 2\right)}=\sqrt{M} \left\{ 1- \frac{\epsilon ^2}8\, \frac{ \MMd}{M^2}
-\frac{\epsilon ^2}{2} \,\frac{ \MXdXd}{ M^2}
+\frac{3\,\epsilon ^2}{8} \,\frac{ (\MXd)^2}{ M^3}
+ \cO(\epsilon^3)\right\}.
\label{eq:RA3}
\eeq
Subtracting \eqref{eq:volDiamond} from \eqref{eq:RA3}, the error in the RA approximation is
\beq
\sqrt{\Sigma\left(-\tfrac M 2\right)}-
\ee{\sqrt{\angl{X}_{0,T}}}
=\frac{\epsilon^2}{2}\,\sqrt{M} \left[ -
\frac{ \MXdXd}{ M^2}
+\frac{3}{4} \,\frac{ (\MXd)^2}{ M^3}
\right]+ \cO(\epsilon^3).
\label{eq:RAcorr}
\eeq
Both terms in \eqref{eq:RAcorr} are proportional to $\rho^2$, so the second-order
correction\footnote{Theorems 2 and 4 of \cite{alos2021difference} compute these in the short-time limit.} to the RA approximation vanishes when $\rho=0$. 
The smile does not determine this correction.  From the expansion of $W_p$
in Section \ref{sec:power}, the three lowest trees enter $W_p$ as
$p\,\MXd+\tfrac12\lambda_p\,\MMd+p^2\,\MXdXd$.  The shift
\[
\MXd\;\to\;\MXd+\delta,\qquad
\MMd\;\to\;\MMd+4\delta,\qquad
\MXdXd\;\to\;\MXdXd-\delta,
\]
realizable within the model class \eqref{eq:SVmodel}, leaves $W_p$ unchanged for every $p$, and with it the
distribution of $S_T$ and the entire maturity-$T$ smile.
The smile determines only the combinations invariant under the shift:
the stochasticity contract $\zeta=\tfrac14\,\MMd-\MXd$ is one such
combination, which is why it admits the magic-strike approximation of
Section \ref{sec:other_claims}.  
Under the shift, $\Sigma\left(-\tfrac M 2\right)$ is unchanged while, to
leading order, $\ee{\sqrt{\angl X_{0,T}}}$ changes by
$-\epsilon^2\delta/(2M^{3/2})$, so no functional of the smile can supply
the correction \eqref{eq:RAcorr}.
In particular, the volatility contract is not attainable.
\section{Numerical tests}\label{sec:numerics}
We now assess the accuracy of the magic-strike approximations in two stochastic volatility models and on SPX option data. In the model-based tests, the figures use 15 evenly spaced maturities between $0.05$ and $1.0$ years.
\subsection{Heston}
We first consider the Heston model,
\[
    d\xi_t(u) = \eta e^{-\lambda (u-t)} \sqrt{\xi_t(t)} \, dW_t,
\]
with initial forward variance curve $\xi_0(u) = \bar V + e^{-\lambda u} (V_0 - \bar V)$. 
For this model, the  variance and gamma contracts are available in closed form:
\begin{align*}
    M(T) &= \bar V T + (V_0 - \bar V) \frac{1 - e^{-\lambda T}}{\lambda}, 
    \\
    G(T) &= \bar V^\prime T + (V_0 - \bar V^\prime) \frac{1 - e^{-\lambda^\prime T}}{\lambda^\prime},
    \qquad \lambda^\prime = \lambda - \rho \eta,
    \quad \lambda^\prime \bar V^\prime = \lambda \bar V.
\end{align*}
\begin{table}[h]
\centering
\begin{tabular}{|c|c|c|c|c|c|}
\hline
 & $V_0$ & $\bar V$ & $\lambda$ & $\eta$ & $\rho$ \\
\hline
Set I \cite[Chapter 6]{bergomi2015stochastic} & $0.16$ & $0.04$ & $1.0$ & $0.6$ & $-0.8$\\
\hline
Set II \cite{bourgey2024smile} & $0.117$ & $0.048$ & $3.37$ & $1.99$ & $-0.68$ \\
\hline
\end{tabular}
\caption{Heston parameter sets used in the numerical tests.}
\label{tab:params_heston}
\end{table}
Table  \ref{tab:params_heston} lists the parameter sets used in our Heston tests; the forward variance curves differ substantially between the test sets, representing very different market environments. 
Figure \ref{fig:heston_magic_strikes} shows the order-5 variance contract magic strikes (see Proposition \ref{prop:M}):
\begin{equation}\label{eq:vs_magic_strikes_5}
    k_0, \quad k_0 \pm \Delta k, \quad k_0 \pm 2\, \Delta k,
\end{equation}
where $k_0=-\widehat M_5/2$, $\Delta k=(\widehat M_5)^{1/2}$, and $\widehat M_5$ is the fifth-order fixed-point estimate of the variance contract value, for $T \in \{0.1,0.25,0.5,1.0\}$ under the two parameter sets in Table \ref{tab:params_heston}. The strikes move further from the money as the volatility-of-volatility and skew increase, but remain well behaved across maturities in both calibrations.
\begin{figure}[H]
    \hspace{-0.5cm}
    \includegraphics[width=0.5\linewidth]{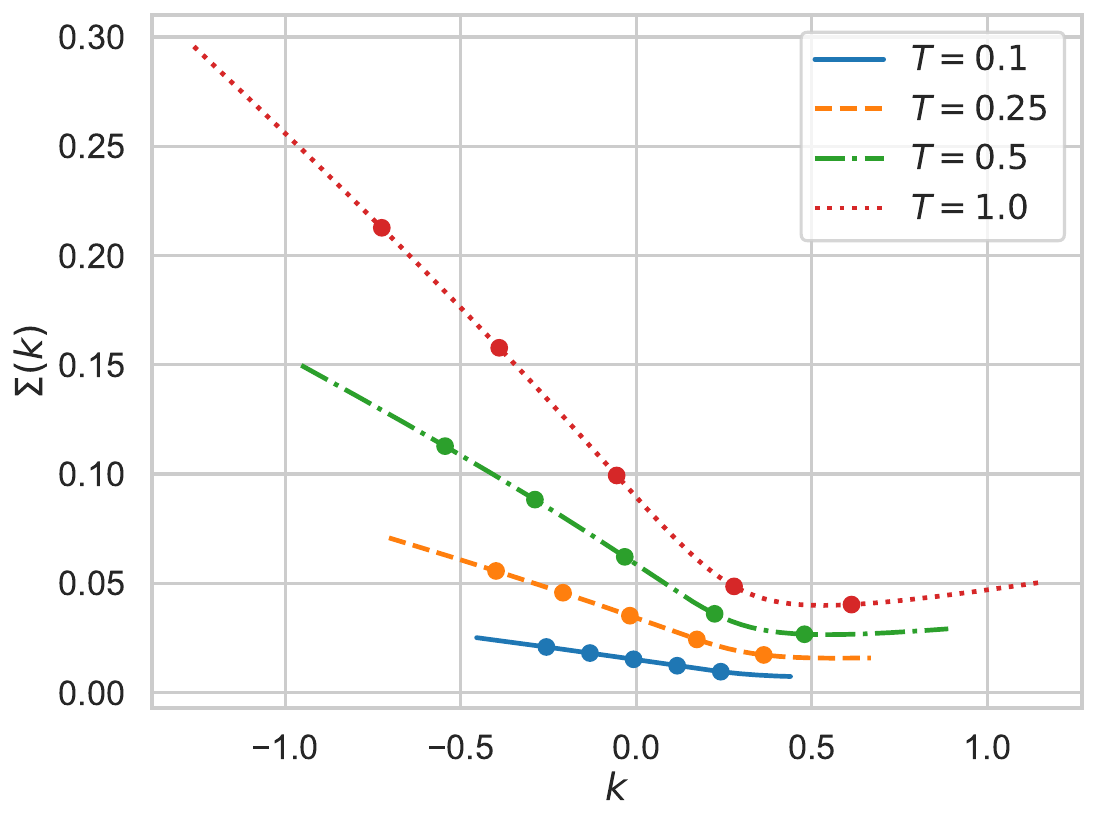}
    \includegraphics[width=0.5\linewidth]{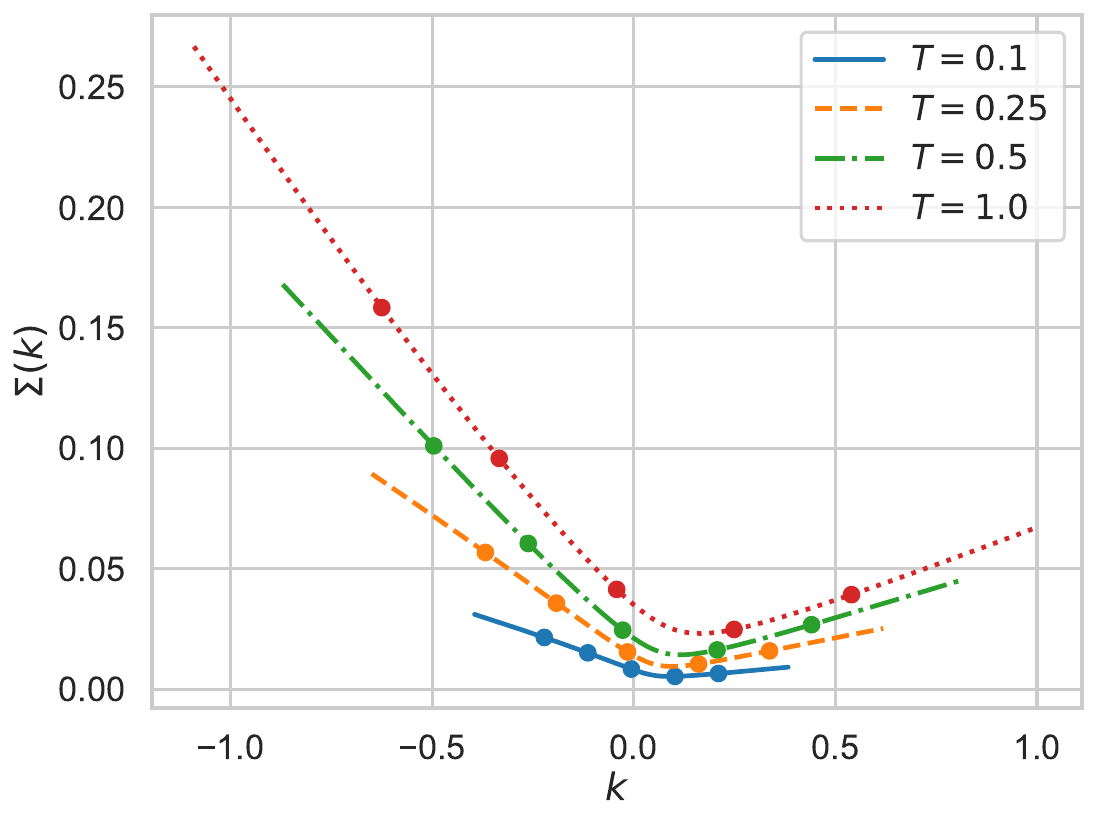}
    \caption{Order-5 variance contract magic strikes \eqref{eq:vs_magic_strikes_5} under the Heston model for the parameter sets in Table \ref{tab:params_heston}: Set I (left) and Set II (right). The maturities are $T \in \{0.1,0.25,0.5,1.0\}$. The curves show total implied variance $\Sigma(k)$; the dots mark the magic strikes.}
    \label{fig:heston_magic_strikes}
\end{figure}

Figure \ref{fig:heston_variance_gamma_swap} reports the order 1--5 approximations for the variance and gamma contracts under the Heston model. In each figure, the panels correspond to Set I (top) and Set II (bottom) in Table \ref{tab:params_heston}. 
For Set I, orders 2--5 improve substantially on the first-order
approximation. Increasing the order does not always improve the estimate,
however. In Set II, orders 2 and 3 are more accurate than orders 4 and 5
over much of the maturity range; the latter overestimate both contract
values. This may reflect how well a low-degree polynomial approximates
the smile as the stencil widens from $\Delta k$ to $2\Delta k$ on either
side of its center.

\begin{figure}[H]
    \centering
    \hspace{-0.5cm}
    \includegraphics[width=0.5\linewidth]{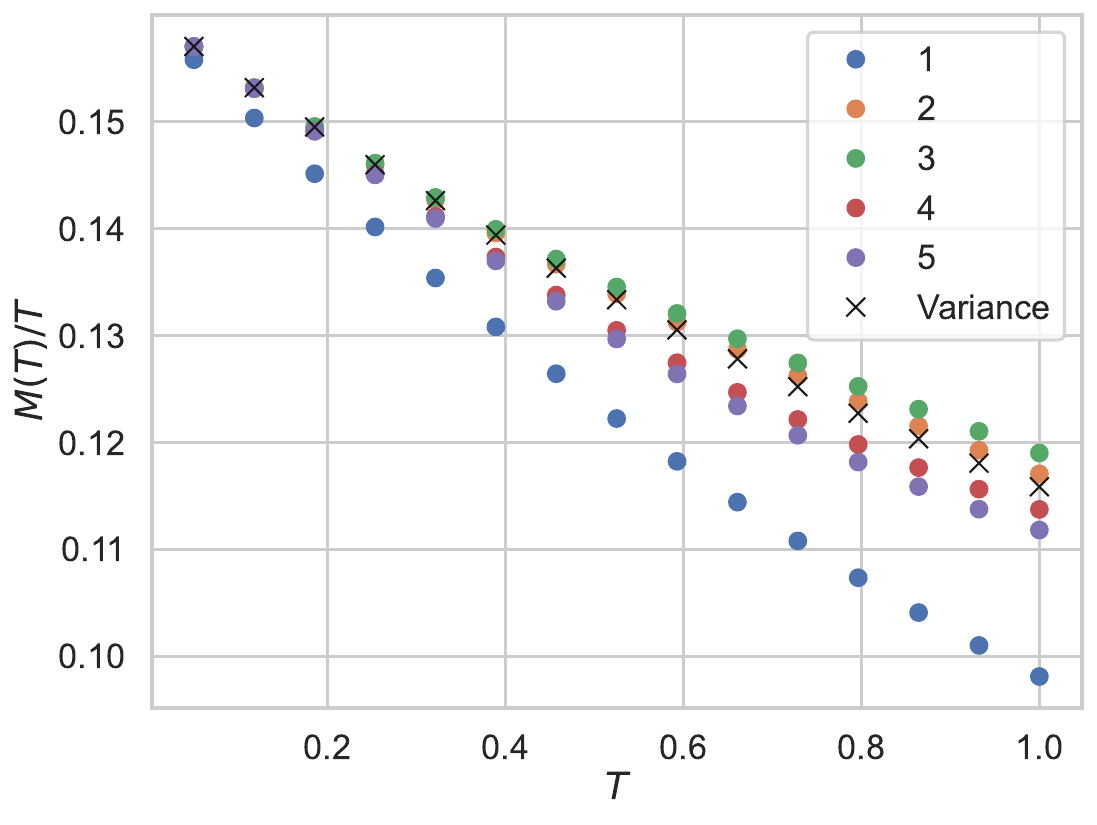}
    \includegraphics[width=0.5\linewidth]{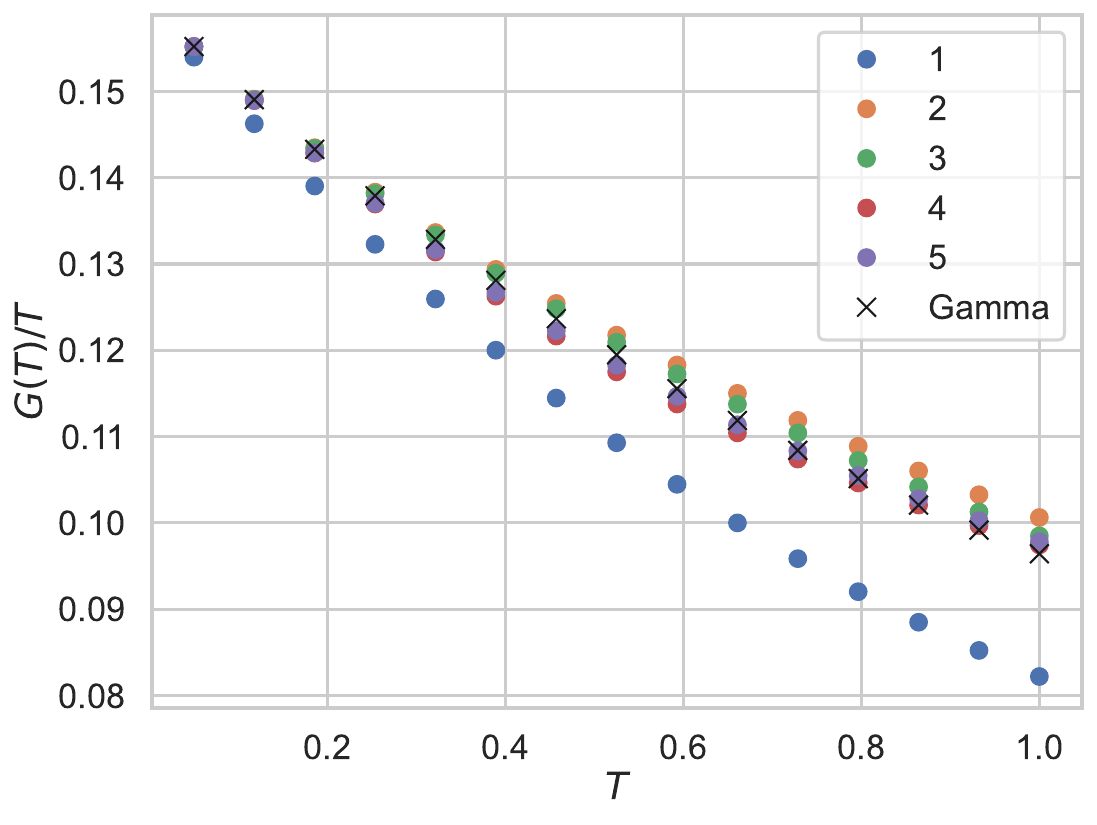}\\
    \hspace{-0.5cm}
    \includegraphics[width=0.5\linewidth]{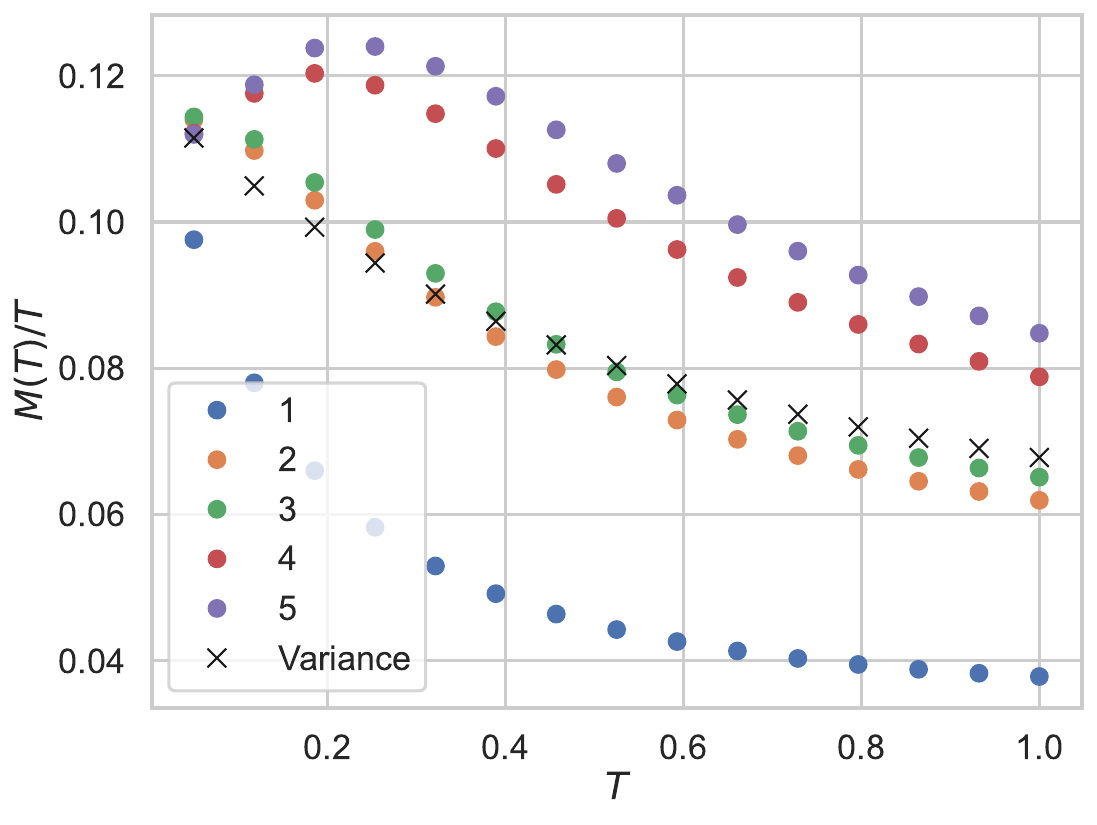}
    \includegraphics[width=0.5\linewidth]{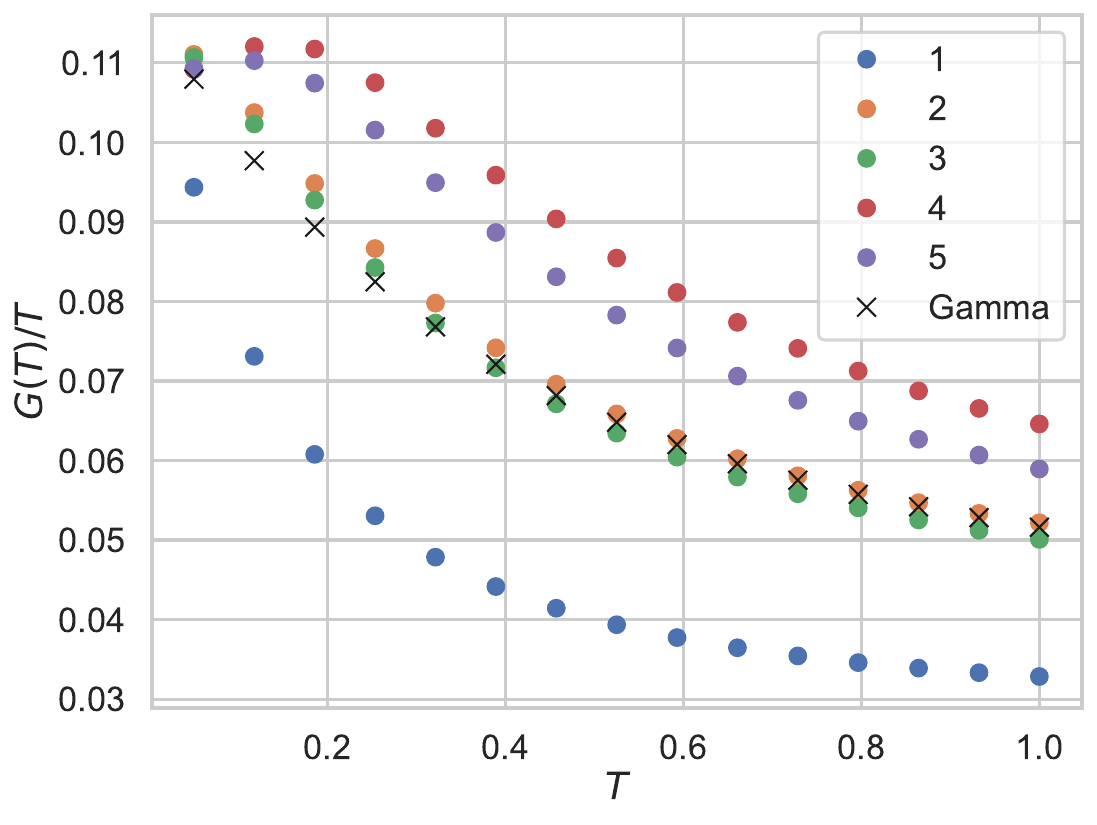}
    \caption{Variance (left) and gamma (right) contract approximations under the Heston model. Set I (top) and Set II (bottom) use the parameters in Table \ref{tab:params_heston}.}
    \label{fig:heston_variance_gamma_swap}
\end{figure}
\subsection{Rough Bergomi}
We next test the approximations in the rough Bergomi model \cite{bayer2016pricing}, where
\begin{equation}
    d\xi_t(u)  = \xi_t(u) \,\eta \,\sqrt{2H}\, (u-t)^{H-1/2} \, dW_t.
\end{equation}
The model is simulated with $300$ time steps and five independent batches
of $3\times10^5$ Monte Carlo paths each, using antithetic variates for
variance reduction. The gamma contract benchmark is estimated by Monte Carlo while the variance contract $M(T)=\int_0^T\xi_0(u)\,du$ is
deterministic. The initial forward variance curves for the two parameter sets represent very different market environments.
\begin{table}[H]
\centering
\begin{tabular}{|c|c|c|c|}
\hline
& $H$ & $\eta$ & $\rho$ \\
\hline
Set I & $0.23$ & $1.5$ & $-0.62$ \\
\hline
Set II & $0.11$ & $2.39$ & $-0.86$ \\
\hline
\end{tabular}
\caption{Rough Bergomi parameters from \cite{bourgey2024smile}.}
\label{tab:params_rough_bergomi}
\end{table}
\begin{figure}[H]
    \hspace{-0.5cm}
    \centering
    \includegraphics[width=0.5\linewidth]{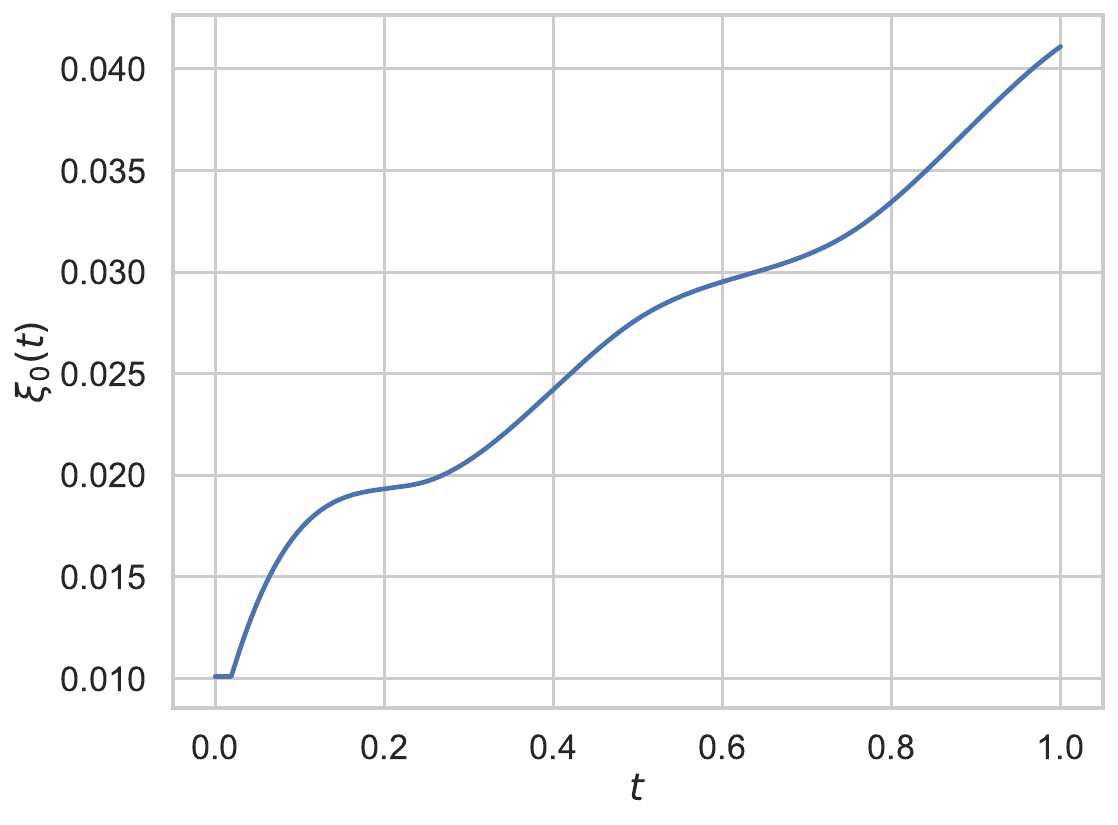}
    \includegraphics[width=0.5\linewidth]{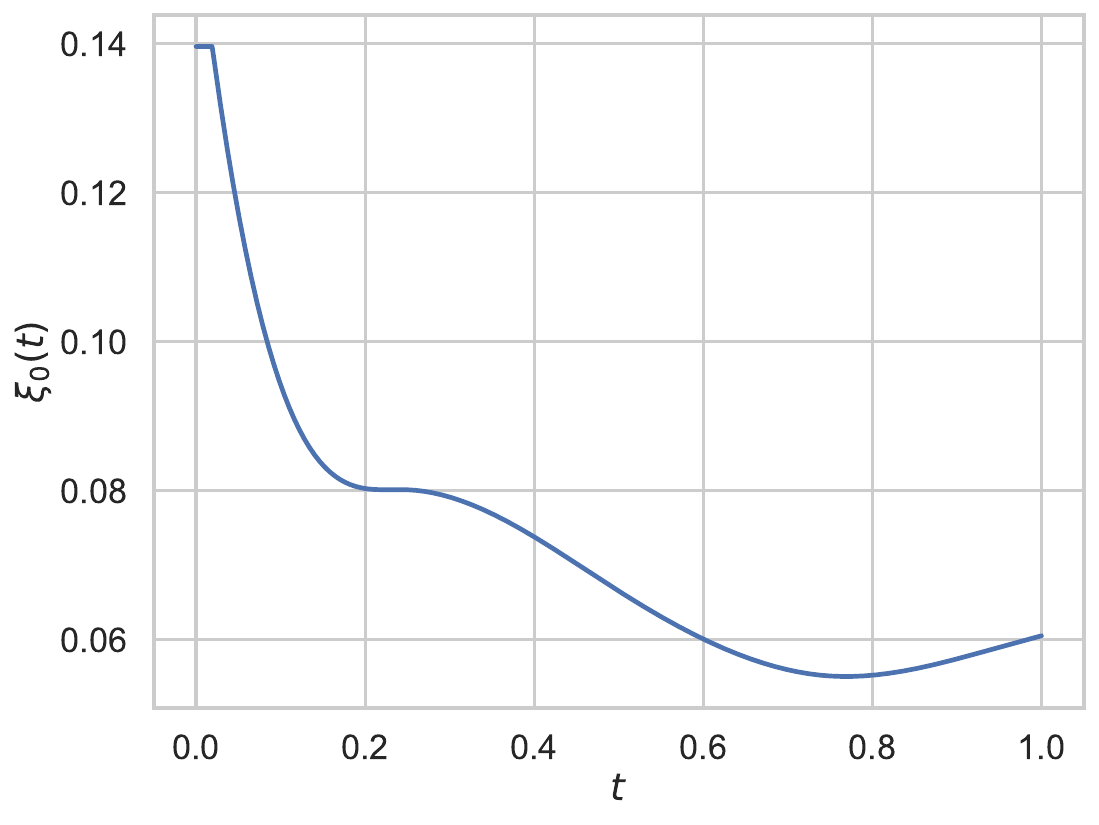}
    \caption{Initial forward variance curve $\xi_0$ for the rough Bergomi parameter sets in Table \ref{tab:params_rough_bergomi}: Set I (left) and Set II (right).}
    \label{fig:rbergomi_xi0}
\end{figure}
Both parameter sets produce pronounced short-maturity skew. Figure \ref{fig:rbergomi_xi0} shows the associated initial forward variance curves, and Figure \ref{fig:rbergomi_variance_gamma_swap} compares the magic-strike approximations for the variance and gamma contracts. 
For both parameter sets, orders 2--5 substantially improve on the
first-order approximation. The preferred order nevertheless depends
on the contract and maturity. For Set II, the fourth-order variance
approximation tracks the benchmark particularly closely, whereas
the second- and third-order gamma approximations are generally closer
to the benchmark than the fourth- and fifth-order estimates.
The latter tend to overestimate the gamma contract value at longer
maturities.

\begin{figure}[H]
    \hspace{-0.5cm}
    \includegraphics[width=0.5\linewidth]{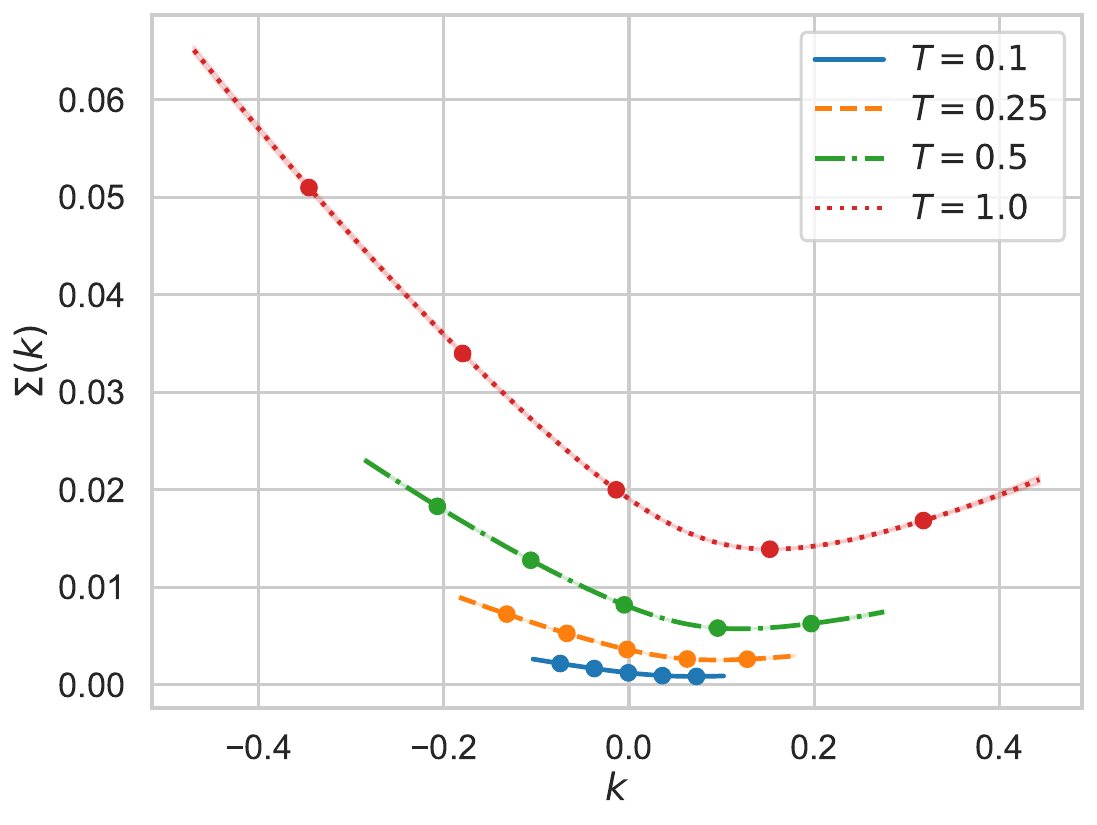}
    \includegraphics[width=0.5\linewidth]{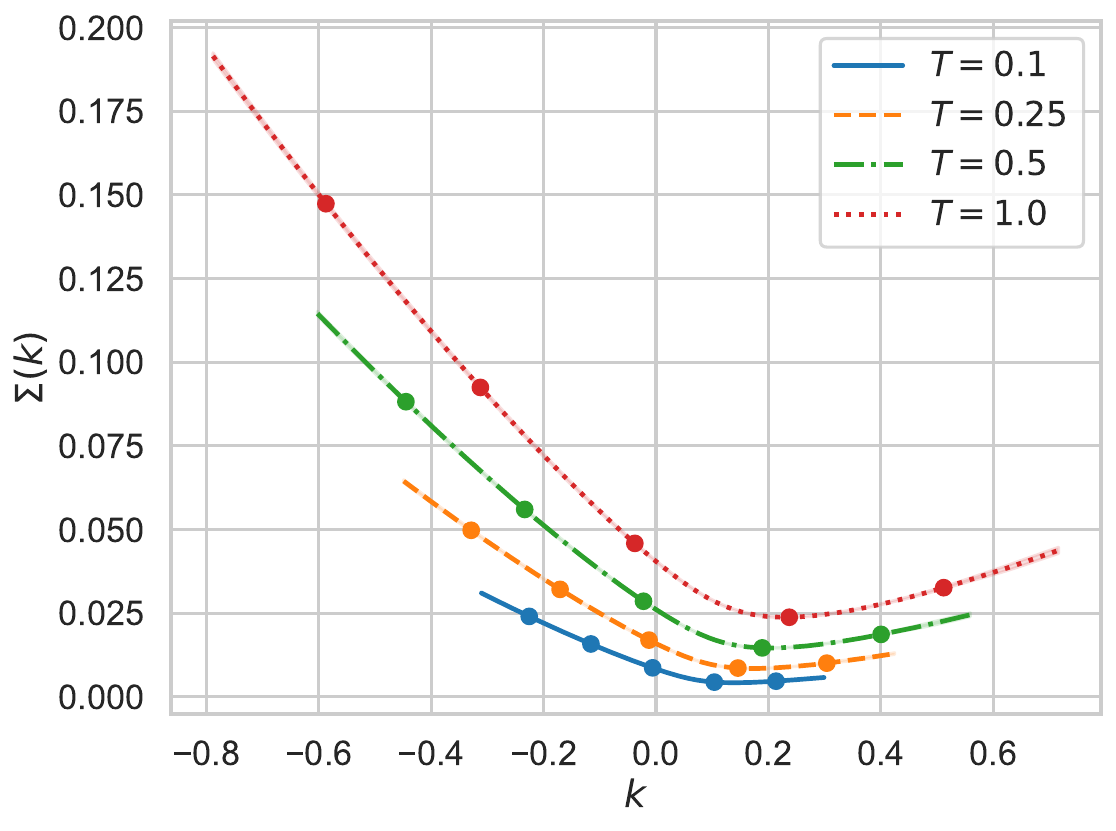}
    \caption{Order-5 variance contract magic strikes \eqref{eq:vs_magic_strikes_5} under the rough Bergomi model for the parameter sets in Table \ref{tab:params_rough_bergomi}: Set I (left) and Set II (right). The maturities are $T \in \{0.1,0.25,0.5,1.0\}$. The curves show total implied variance $\Sigma(k)$; the dots mark the magic strikes.}
    \label{fig:rbergomi_magic_strikes}
\end{figure}

\begin{figure}[H]
    \centering
    \hspace{-0.5cm}
    \includegraphics[width=0.5\linewidth]{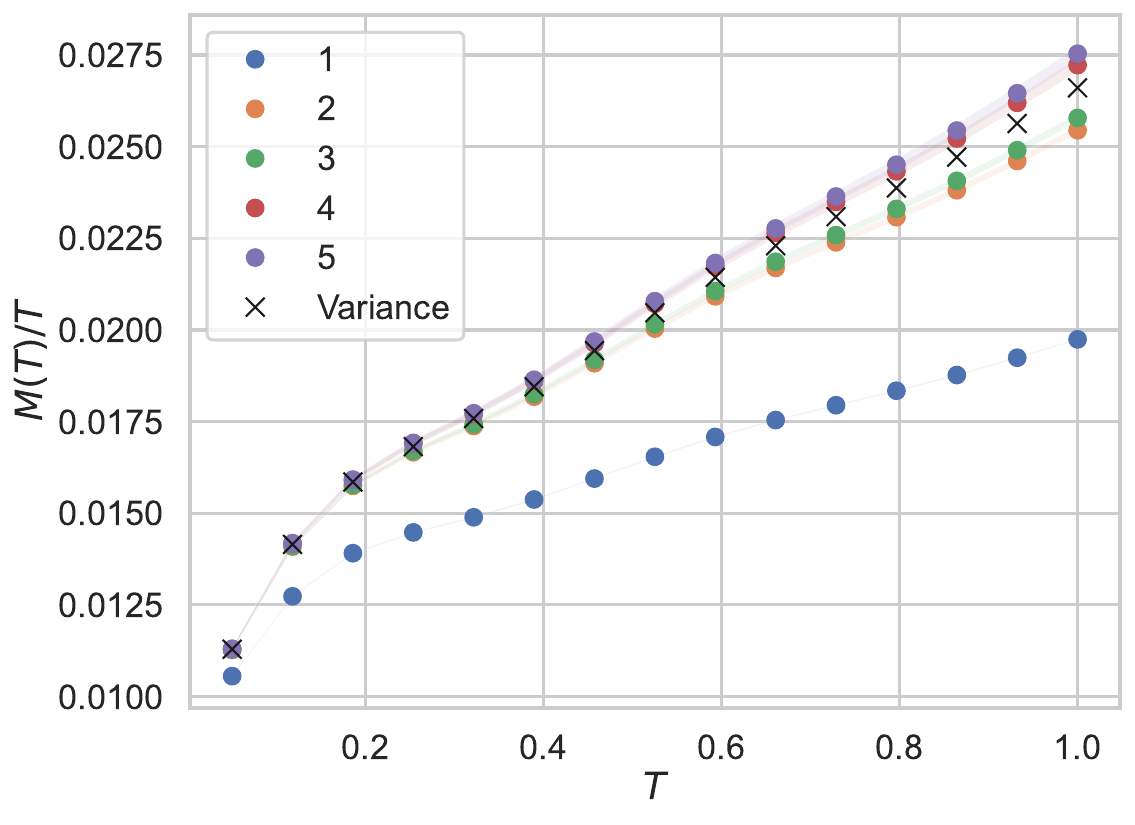}
    \includegraphics[width=0.5\linewidth]{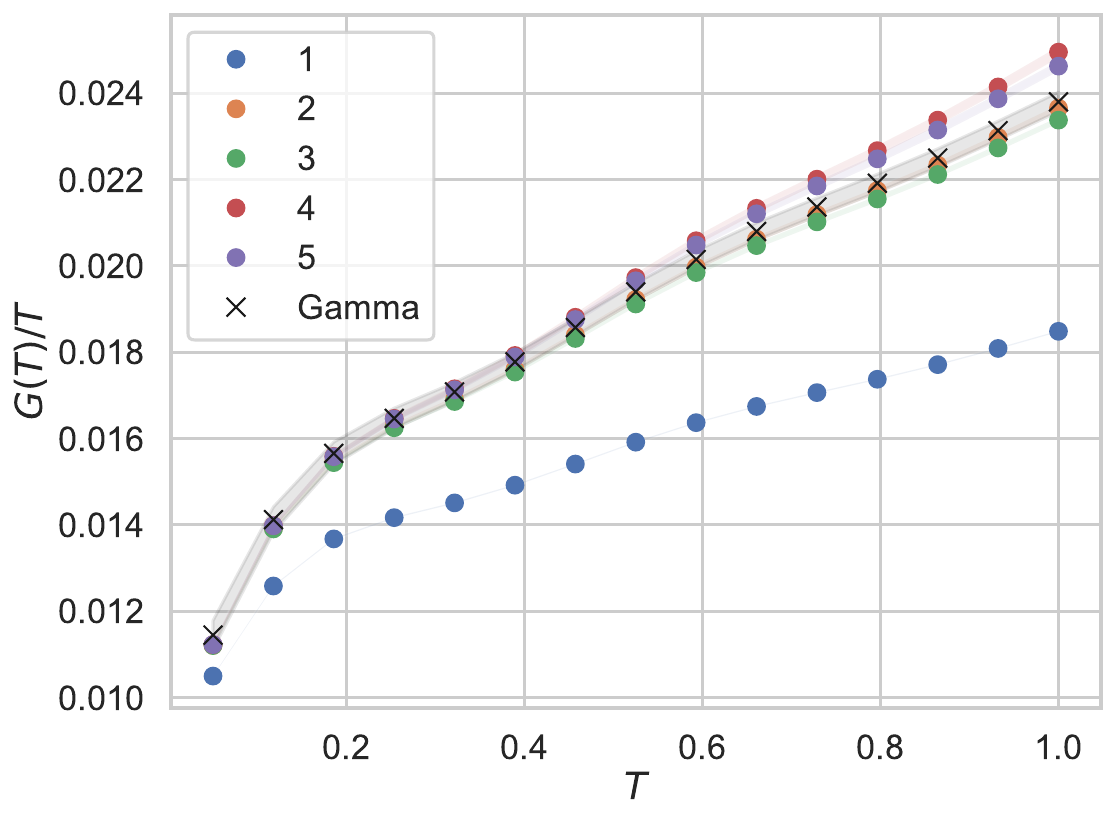}\\
    \hspace{-0.5cm}
    \includegraphics[width=0.5\linewidth]{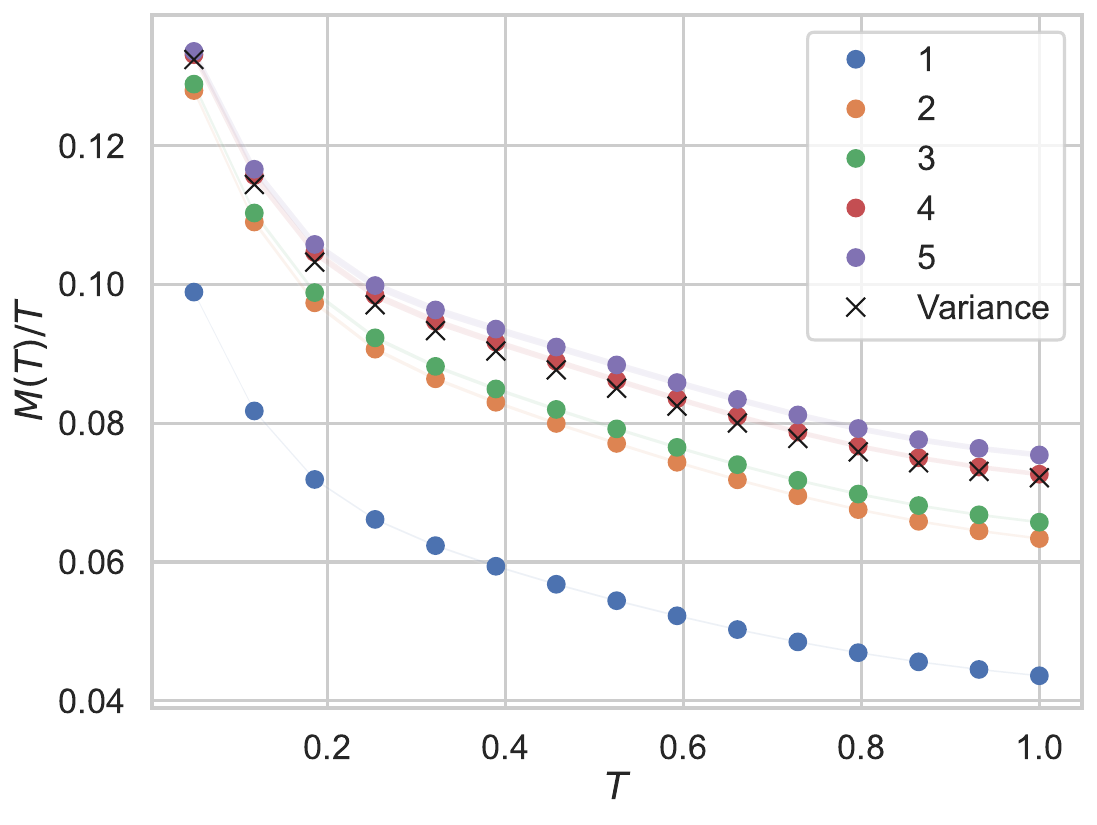}
    \includegraphics[width=0.5\linewidth]{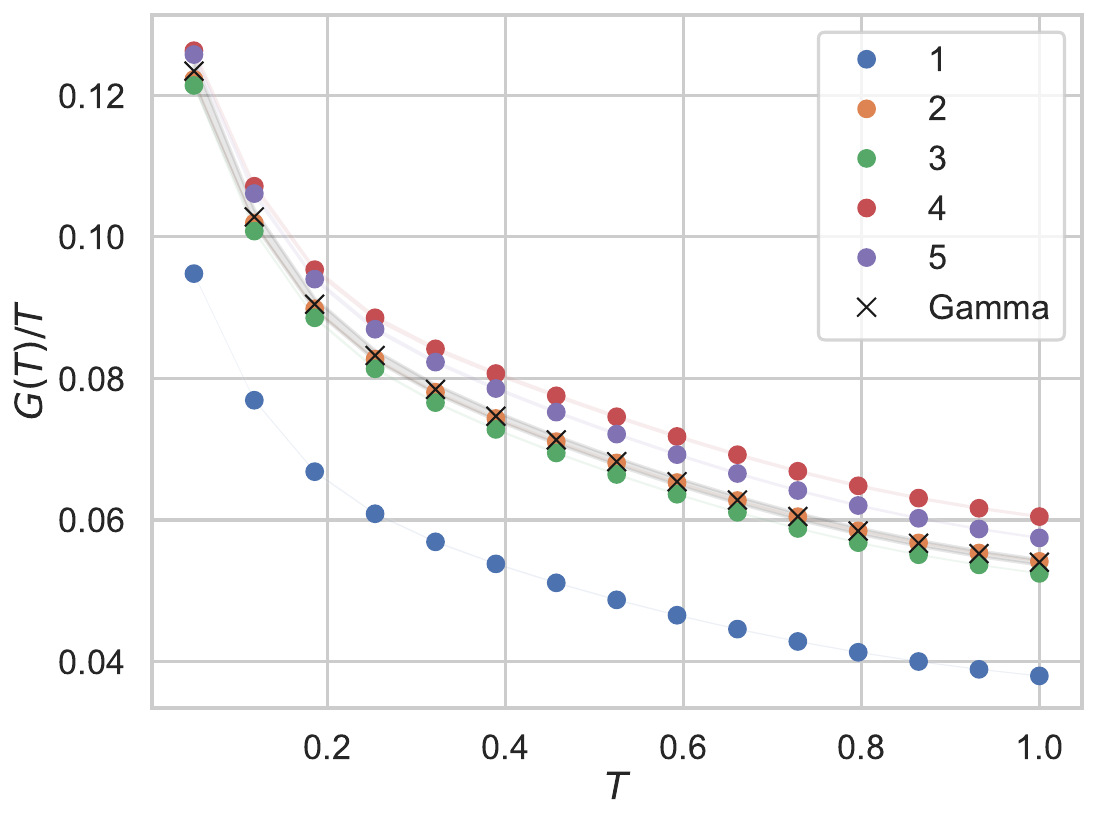}
    \caption{Variance (left) and gamma (right) contract approximations under the rough Bergomi model. Set I (top) and Set II (bottom) use the parameters in Table \ref{tab:params_rough_bergomi}. Monte Carlo confidence intervals computed from five independent batches are also shown.}
    \label{fig:rbergomi_variance_gamma_swap}
\end{figure}
\subsection{The market}
Finally, we apply the method to SPX implied volatility surfaces on 10
February 2025, 4 April 2025, 7 April 2025, and 2 July 2025. The smile is
interpolated on each maturity slice and extrapolated flat outside the
observed log-moneyness range.
For the variance and gamma contracts, the benchmarks are the Fukasawa
representations \eqref{eq:VS} and \eqref{eq:GS} evaluated on the same
interpolated smile with $10$ quadrature nodes.
The extrapolation affects only these benchmarks: on every slice shown,
the magic strikes lie within the quoted strike range, so the
magic-strike estimates are independent of the extrapolation choice.

Figure \ref{fig:spx_market_variance_gamma_swap} shows substantial improvements
over the first-order approximation. On 10 February and 2 July, orders 4 and
5 are closer to the plotted Fukasawa values than orders 2 and 3 for
variance, but lie above the Fukasawa estimates for gamma. On the two April dates,
orders 2--5 are more closely grouped. The Vola Dynamics\footnote{We are very grateful to Vola Dynamics for
   providing their variance and gamma contract marks.} marks shown for 2 July 2025 are also in close
   agreement with the magic-strike estimates.
The computation remains inexpensive. On the same data set, the runtimes (on
a MacBook Air M5, 24 GB RAM) in seconds were $0.022$ for the Fukasawa
benchmark and $0.037$, $0.049$, $0.070$, $0.14$, and $0.19$ for orders 1
through 5, respectively.
\newpage
\begin{figure}[H]
    \vspace{-2cm}
    \centering
    \hspace{-0.5cm}
    \includegraphics[width=0.5\linewidth]{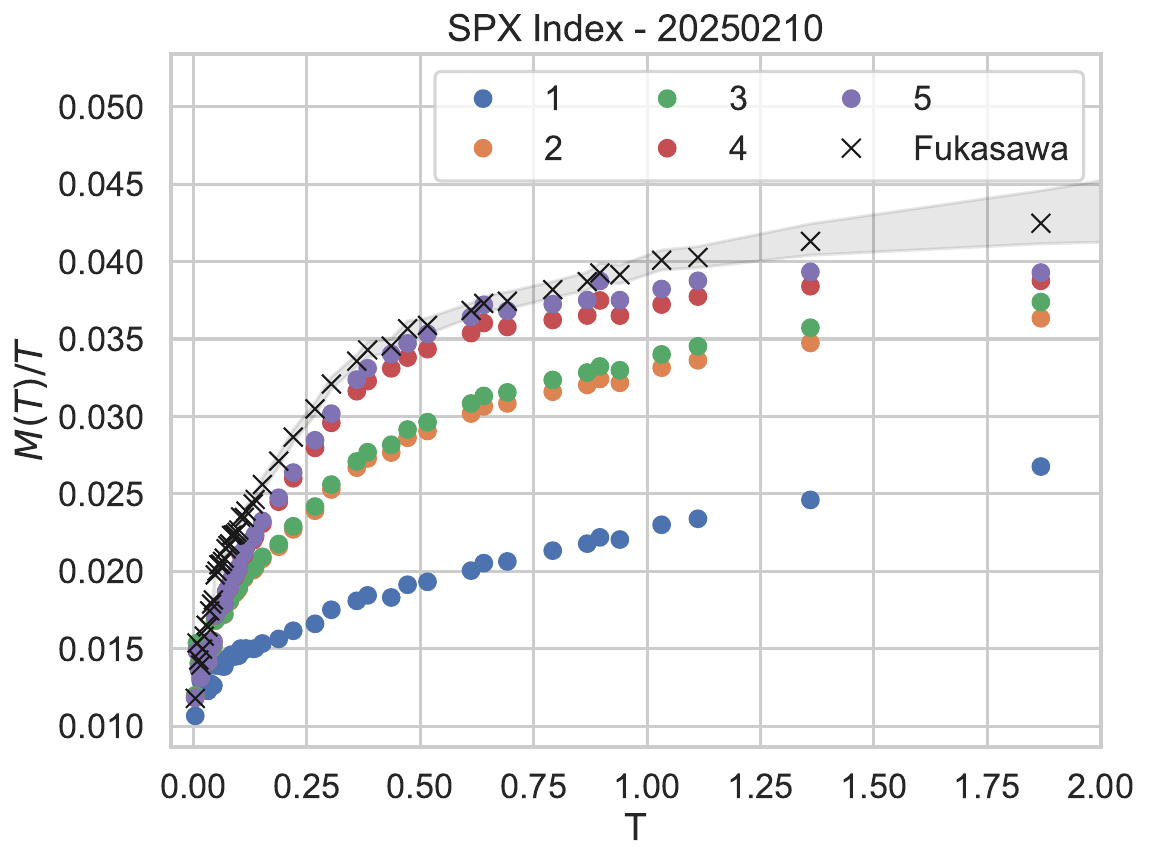}
    \includegraphics[width=0.5\linewidth]{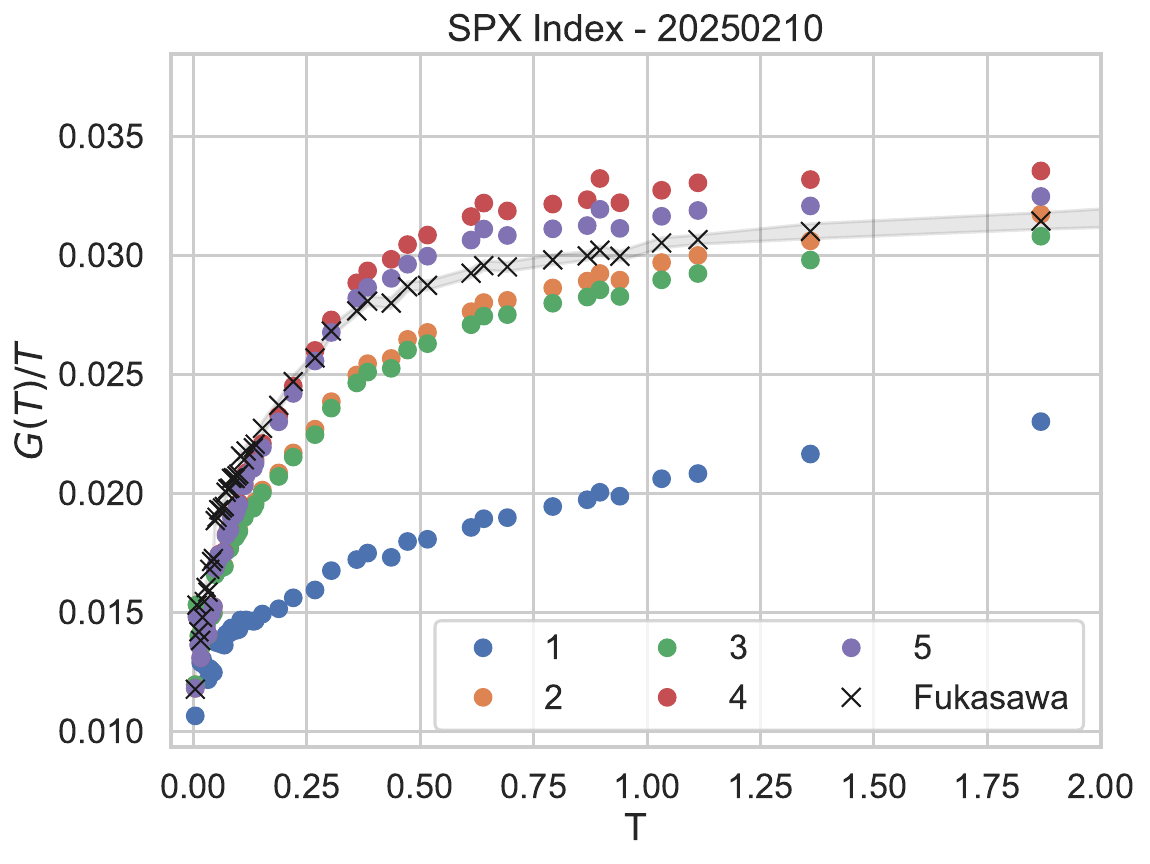}
    \\
    \hspace{-0.5cm}
    \includegraphics[width=0.5\linewidth]{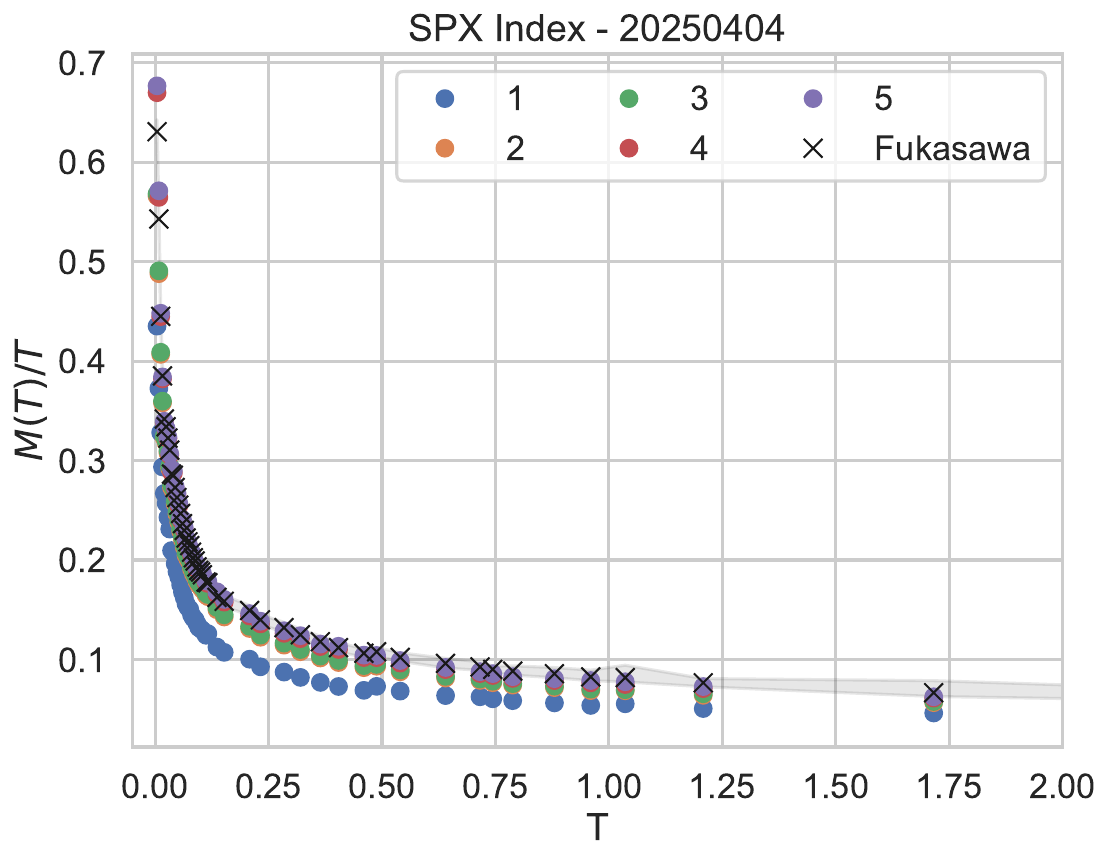}
    \includegraphics[width=0.5\linewidth]{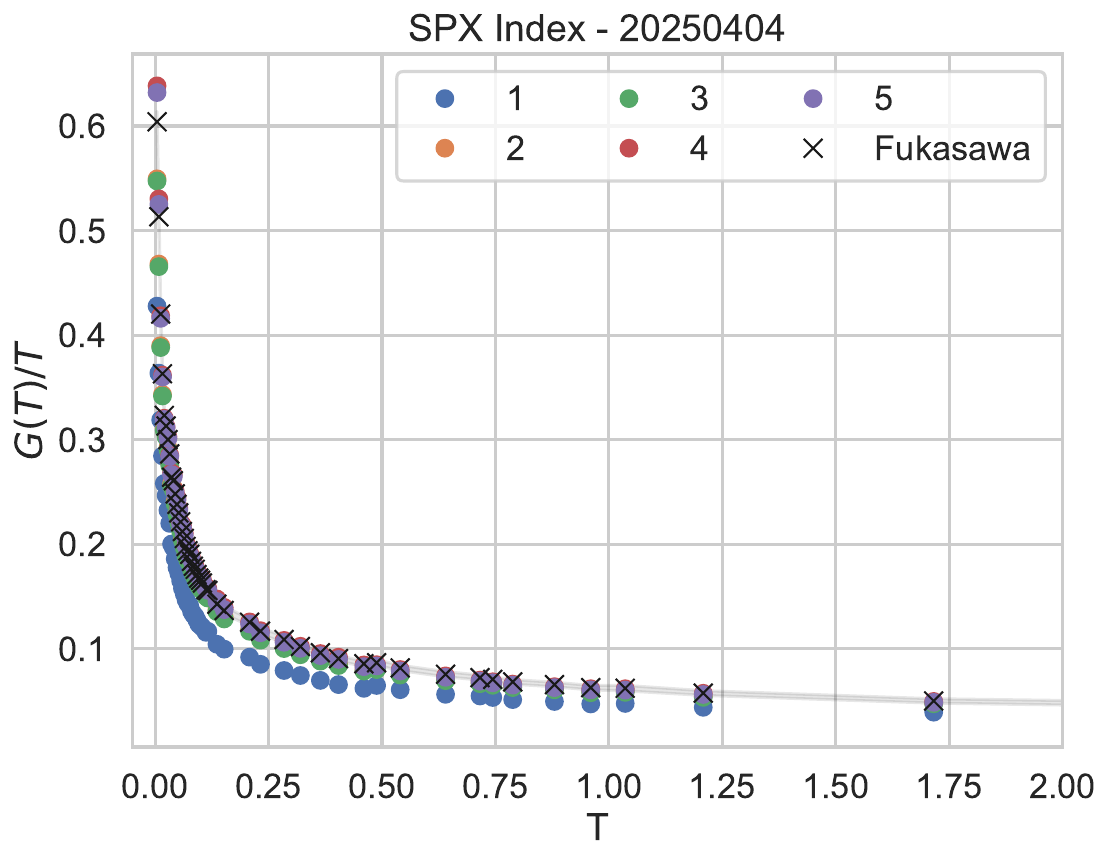}
    \\
    \hspace{-0.5cm}
    \includegraphics[width=0.5\linewidth]{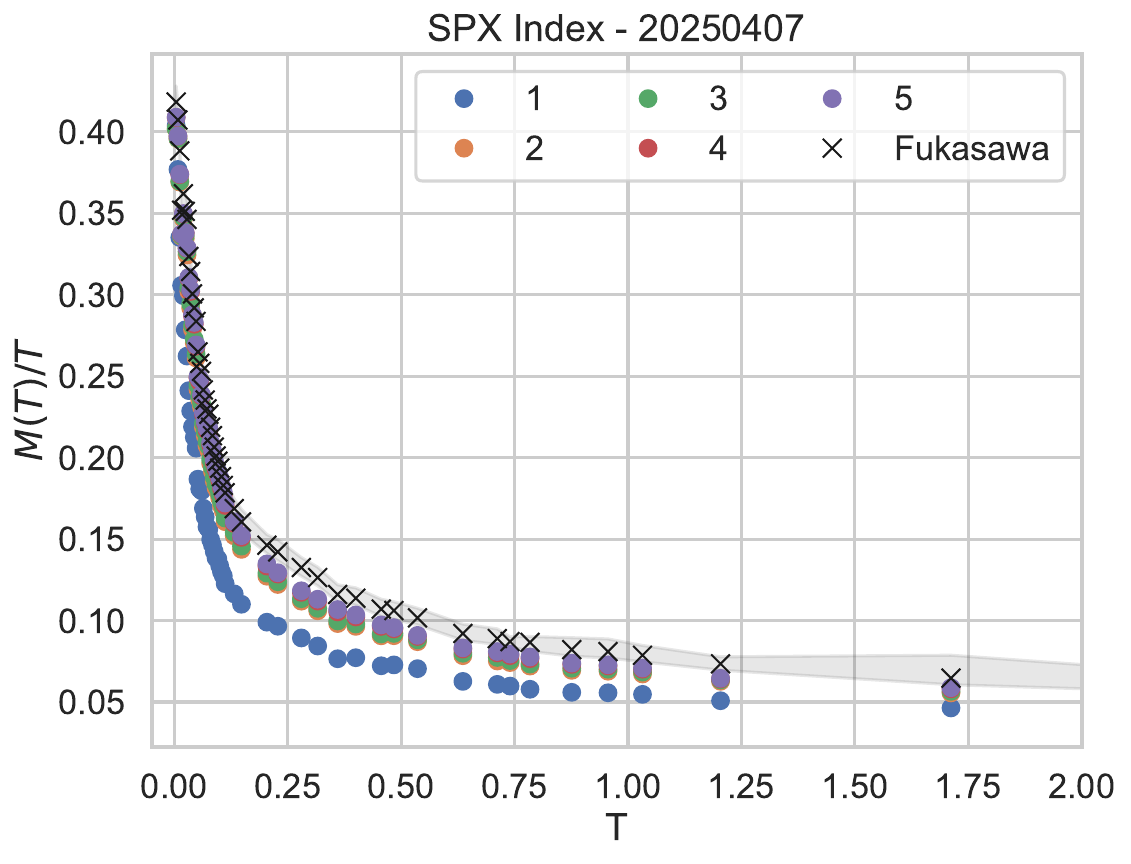}
    \includegraphics[width=0.5\linewidth]{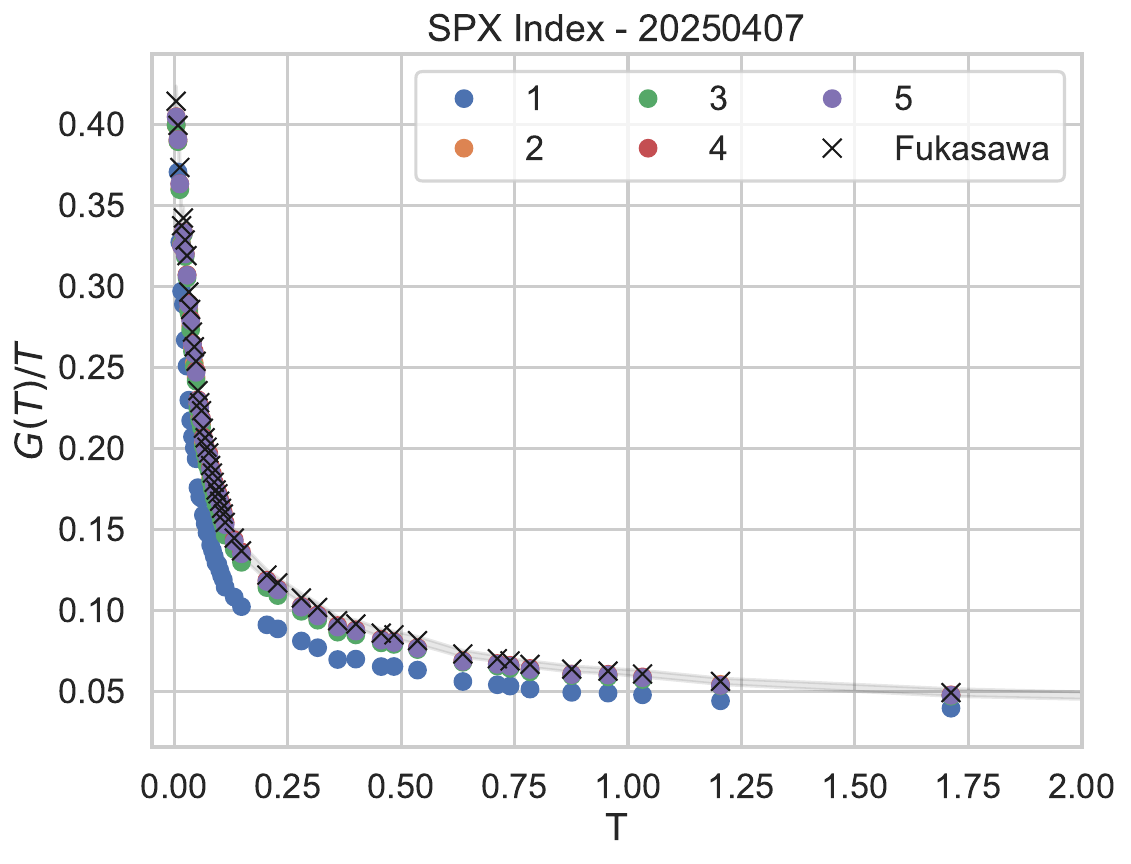}
    \\
    \hspace{-0.5cm}
    \includegraphics[width=0.5\linewidth]{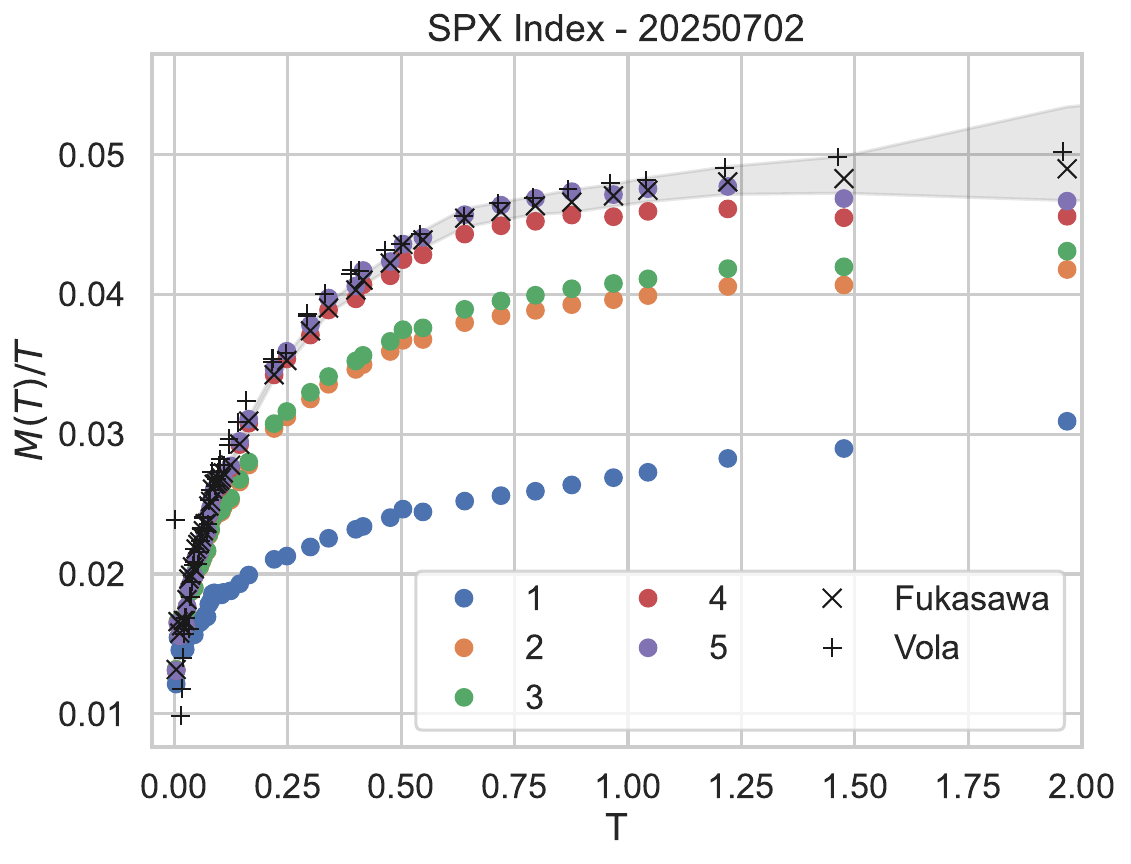}
    \includegraphics[width=0.5\linewidth]{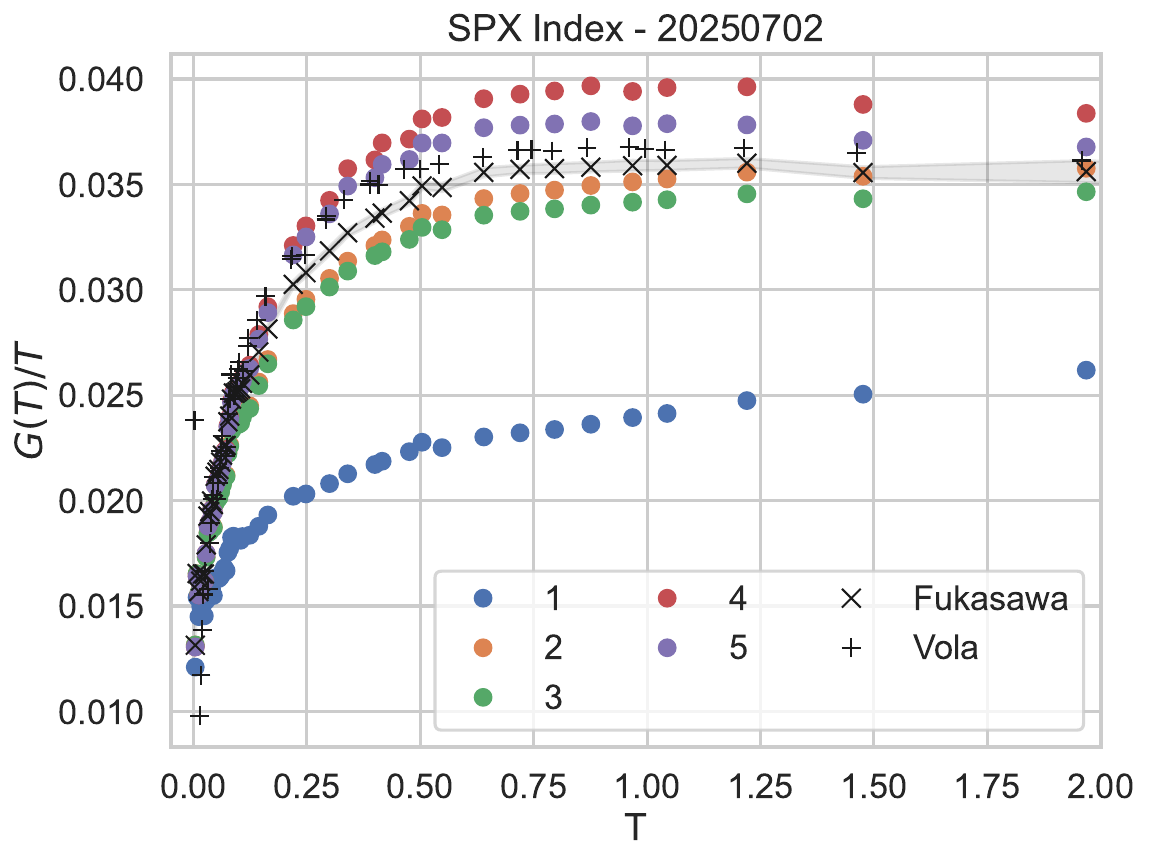}
    \caption{Annualized SPX variance (left) and gamma (right) contract estimates, $M(T)/T$ and $G(T)/T$, for 10 February, 4 April, 7 April, and 2 July 2025 (top to bottom). Colored dots show magic-strike approximations at orders 1--5 using mid smiles; black crosses show the Fukasawa estimates. Gray shading spans the Fukasawa estimates from bid and ask smiles. Black plus signs show Vola Dynamics marks for 2 July. The smile is extrapolated flat outside the observed strike range.}
    \label{fig:spx_market_variance_gamma_swap}
\end{figure}

\section{Summary}\label{sec:conclusion}
Using the forest expansion and the smile expansion of
\cite{bourgey2026demystifying}, we have extended the single-strike
approximations \eqref{eq:approx0} of the variance and gamma contracts to
higher order, and more generally to the one-parameter family of power
payoffs connecting them, deriving model-independent fixed-point
approximations up to fifth order and proving that such approximations
exist at every order.  Numerical tests under Heston and rough Bergomi show
good accuracy even for strongly skewed smiles, and on SPX market data the
method tracks the Fukasawa benchmark.
For the volatility contract, the Rolloos--Arslan approximation is the
square root of the first-order magic-strike approximation of the variance
contract; its second-order correction vanishes when $\rho=0$ but is
otherwise not determined by the smile.
Because the smile is read only at the magic strikes, the approximations
are independent of the smile beyond two standard deviations of the
money, and in particular of the extrapolation choice whenever the magic
strikes lie within the quoted range.  What the construction does require
is that the smile be locally polynomial: its accuracy is governed by the
quality of the polynomial fit at the magic strikes.

\section*{AI Use Disclosure} The authors made extensive use of generative artificial intelligence tools (Claude Opus 5 and Claude Fable 5) throughout the preparation of this manuscript. These tools were used to assist with writing and editing the manuscript, exploring mathematical ideas, proposing proof strategies, checking intermediate arguments, and suggesting improvements to existing proofs. The AI tools were used solely as assistants; all mathematical arguments were developed, critically assessed, and independently verified by the authors. The authors accept full responsibility for the correctness and originality of all results presented.

\bibliographystyle{alpha}
\bibliography{MagicStrikes}

\end{document}